\documentclass[12pt,a4paper,titlepage, headinclude]{scrreprt}

\usepackage{tocloft}
\usepackage[utf8]{inputenc}

\usepackage[ngerman, english]{babel}

\usepackage{setspace} 
\usepackage[symbol]{footmisc}

\usepackage[backend=biber,bibencoding=utf8]{biblatex}

\usepackage{scrlayer-scrpage}
\clearpairofpagestyles
\automark[section]{section}
\ohead{\headmark}
\KOMAoptions{
	headsepline = true}
\usepackage{enumitem}\setitemize{noitemsep,topsep=0pt,parsep=0pt,partopsep=0pt}
	\setenumerate{noitemsep,topsep=5pt,parsep=3pt,partopsep=20pt}

\usepackage{latexsym,exscale,stmaryrd,amssymb,amsmath,amsthm, mathrsfs}
\usepackage[mathscr]{eucal}

\makeatletter
\newcommand{\dalemb}{\mathop{\mathpalette\dalemb@\relax}}
\newcommand{\dalemb@}[2]{%
	\begingroup
	\sbox\z@{$\m@th#1\square$}%
	\dimen0=\fontdimen8
	\ifx#1\displaystyle\textfont\else
	\ifx#1\textstyle\textfont\else
	\ifx#1\scriptstyle\scriptfont\else
	\scriptscriptfont\fi\fi\fi3
	\makebox[\wd\z@]{%
		\hbox to \ht\z@{%
			\vrule width \dimen0
			\kern-\dimen0
			\vbox to \ht\z@{
				\hrule height \dimen0 width \ht\z@
				\vss
				\hrule height 1.5\dimen0
			}%
			\kern-2\dimen0
			\vrule width 2\dimen0
		}%
	}%
	\endgroup
}
\makeatother

\theoremstyle{plain}
\newtheorem{theorem}{Theorem}[section]

\theoremstyle{plain}
\newtheorem{lemma}[theorem]{Lemma}
\theoremstyle{plain}
\newtheorem{corollary}[theorem]{Corollary}
\theoremstyle{plain}
\newtheorem{proposition}[theorem]{Proposition}

\theoremstyle{definition}
\newtheorem{definition}[theorem]{Definition}

\theoremstyle{remark}
\newtheorem*{remark}{Remark}

\theoremstyle{plain}
\newtheorem{manualtheoreminner}{Theorem}
\newenvironment{manualtheorem}[1]{%
	\renewcommand\themanualtheoreminner{#1}%
	\manualtheoreminner
}{\endmanualtheoreminner}

\theoremstyle{definition}
\newtheorem{manualdefinitioninner}{Definition}
\newenvironment{manualdefinition}[1]{%
	\renewcommand\themanualdefinitioninner{#1}%
	\manualdefinitioninner
}{\endmanualdefinitioninner}

\usepackage{color}

\usepackage{graphicx}
\usepackage{epstopdf}

\usepackage{listings}

\usepackage[hidelinks]{hyperref}

\usepackage[compat=1.1.0]{tikz-feynman}

\begin{document}


	\begin{titlepage}

\begin{centering}

{\huge
  \begin{singlespace}
  The Master Ward Identity\\
  for the Complex Scalar Field
	\footnote[1]{Bachelor's Thesis completed at the Institute for Theoretical Physics on May 25, 2020.}
\end{singlespace}
}

\vspace*{2cm}

\large
\begin{onehalfspace}
	Luis Peters \footnote[2]{Email: \texttt{luis.peters@stud.uni-goettingen.de}}\\
\end{onehalfspace}
Institute for Theoretical Physics\\
Georg-August-Universität Göttingen\\
Friedrich-Hund-Platz 1, 37077 Göttingen, Germany \\
\onehalfspacing
March 5, 2021

\vspace{4cm}

\end{centering}

	\begin{center}
		\textbf{Abstract}
	\end{center}
	
	The Master Ward Identity (MWI) gives a universal formulation of the symmetries of a classical field theory. It is a renormalization condition for the time ordered products of the corresponding quantum field theory. We show that the MWI for a complex scalar field with quartic interaction can be satisfied, with the current, the interaction and all their submonomials as allowed arguments. The proof is performed in the framework of deformation quantization combined with causal perturbation theory, which is summarized and introduced. Some examples of Ward Identities following from the proven MWI are given.

\end{titlepage}

	\tableofcontents
	\addtocontents{toc}{\protect\thispagestyle{empty}}
	
	\newpage
	
	\pagenumbering{arabic}
	

	\addchap{Introduction}
	\label{sec:introduction}
	\clearpairofpagestyles
\ohead{Introduction}
\KOMAoptions{
	headsepline = true}
\cfoot{\pagemark}

Quantum field theory is a framework to describe the physics of quantum matter, which was mainly developed for the purposes of particle physics, that is phenomena taking place at small scales and high energies. Therefore QFT amounts to combining elements of quantum mechanics with the principle of relativity. From an experimental point of view, the accuracy of the standard model formulated in terms of perturbative QFT gives a clear demonstration of the descriptive power of quantum field theory -- especially in quantum electrodynamics. However, from a theoretical perspective fundamental questions remain open: Establishing the Standard Model -- or even one of its subtheories -- as a mathematically complete and consistent theory remains an unsettled challenge \cite{fredrehr}. These theoretical shortcomings lead to other approaches to QFT than the perturbative, namely axiomatic ones, that do not start from the phenomena to be described, but from general principles every QFT should satisfy and then proceed to construct such theories. This in turn may lead to models that do not describe any of the physical phenomena in our world, or rather special cases like restrictions to lower space dimensions. The approach taken in this thesis can be considered as an intermediate one, involving both perturbative an axiomatic elements. \\

One of the main issues of perturbative quantum field theory in its textbook formulation is the appearance of divergent quantities \cite[Chap. 14]{sred07}. They can be classified into three different kinds. First, IR divergences occur due to wrong assumptions on the region where interactions take part. We will address this by localizing interactions in a compact domain of spacetime. The second kind of divergences arises in the UV regime, where local interaction require to take products of functionals at a point, which are not a priori well defined. These quantities will have to be renormalized in an appropriate way, namely in the framework of causal perturbation theory. The last type of divergences concerns the overall sum of the perturbative expansions. We will not fix this issue, that is we deal with formal power series where no convergence is implied. \\

The purpose of our formulation of QFT is to clarify the connection between classical and quantum symmetries, like the conservation of certain currents. Many textbooks give these relations in terms of so called Ward Identities \cite[Chap. 7.4]{peskin95}. The viewpoint we take is to formulate classical symmetries in terms of the Master Ward Identity (MWI), which we impose as an additional condition on the quantum theory, whereby it is not clear whether this can always be satisfied. The aim of this thesis is to show that the MWI can be satisfied for a complex scalar field with quartic interaction, adopting a proof given by Michael Dütsch and Klaus Fredenhagen for the case of QED \cite{duetschfred99}. \\

The guiding principle of our approach is to consider a QFT as a modified (quantized) version of a classical theory, where we want to maintain as many of the structural properties of the classical theory as possible. This general framework is described in chapter \ref{chap:found}, where we mainly follow \cite[chap. 1-3]{duetsch19}. The aim is to give a concise outline of the approach we use, introducing all main objects and stating the results which will be needed for our proof of the MWI. We will shortly point out the physical ideas behind the basic definitions, but won't perform any comprehensive discussions. Starting point is the introduction of a framework of classical scalar field theory in section \ref{sec:clft}  which is well suited for the transition to quantum fields. Fields are described as a certain class of functionals on configuration space, that form a commutative $*$-algebra, on which a Poisson structure is introduced that encodes the free dynamics. Interactions are then described perturbatively in terms of free fields. In section \ref{sec:defquant} the free classical theory is quantized by deforming the commutative classical algebra into a non-commutative quantum algebra of formal power series in $\hbar$. Interactions are reintroduced into the quantum theory in section \ref{sec:interact} by using causal perturbation theory, which amounts to understanding the process of renormalization as the extension of distributions to certain points. This completes the construction of an interacting QFT. \\

Chapter \ref{chap:symms} introduces the framework of the MWI as a universal formulation of symmetries. In section \ref{sec:symms} the MWI is derived from the properties of classical fields and its relation to Noether's theorem is discussed. Then we clarify its status as a renormalization condition for interacting quantum fields and present the anomalous MWI, which gives a characterization of the possible violations of the MWI that can occur in the quantized theory. \\

Chapter \ref{chap:proof} contains our main original investigations. It is shown that for a complex scalar quantum field with quartic interaction the relevant MWI corresponding to the global $U(1)$ symmetry can always be satisfied. Section \ref{sec:prelims} presents four statements that will be required and proves two of them. In section \ref{sec:proofMWI} the actual proof is performed. It proceeds by induction on the number of arguments of the anomalous map and shows that the anomaly can be removed at every order by an admissible finite renormalization. The last section \ref{sec:examples} concludes the thesis by calculating some exemplary Ward Identities that follow from the MWI we have shown.

\clearpairofpagestyles
\automark[section]{section}
\ohead{\headmark}
\KOMAoptions{
	headsepline = true}
\cfoot{\pagemark}

	\chapter{Foundations: From classical to quantum field theory}
	\label{chap:found}
	
	\section{Classical Field Theory}
	\label{sec:clft}
	This section introduces the framework of classical field theory and its perturbative formulation. The space of free fields consists of functionals on configuration space, endowed with the structure of a commutative Poisson $*$-algebra. Interacting fields can be described as formal power series in the coupling constant. Our goal is to give a formulation of classical theory that carries over nicely into the quantum world. 

\subsection{Kinematics} 

We throughout will be describing the case of a single scalar field to keep the expressions as simple as possible. Hence we take the configurations of the field $\phi$ to be described by smooth real-valued functions on $d$-dimensional Minkowski space $\mathbb{M}_d\equiv\mathbb{M}$. The complex case is introduced in section \ref{sec:complex}.

\begin{definition}
	The \textit{configuration space} of real scalar theory is the space 
	\begin{equation*}
	\mathcal{E} := \mathcal{C}^{\infty} (\mathbb{M},\mathbb{R}) \; .
	\end{equation*}
\end{definition}

\begin{definition}
	The \textit{basic field} $\varphi(x)$ is the evaluation functional at $x$ on the confi-guration space
	\begin{equation*}
		\!
		\begin{aligned}
			\varphi: \mathbb{M} &\rightarrow \mathcal{E}' \\
			x &\rightarrow \varphi(x)
		\end{aligned}
		\qquad \qquad
		\text{and}
		\qquad \qquad
		\!
		\begin{aligned}
			\varphi(x): \mathcal{E} &\rightarrow \mathbb{R} \\
			h &\rightarrow h(x) \; .
		\end{aligned}
	\end{equation*}
	The \textit{partial derivatives} of the basic fields are defined as
	\begin{align*}
		\partial^a \varphi(x): \mathcal{E} &\rightarrow \mathbb{R} \\
		h &\rightarrow \partial^a h(x) \; ,
	\end{align*}
	where $a = (a_1, \dots, a_d) \in \mathbb{N}^d$ is a multi-index. A \textit{general field} is a function $F$ of the basic field $\varphi$, where $F(\varphi)$ is a $\mathbb{C}$-valued functional on configuration space: 
	\begin{align*}
		F(\varphi): \mathcal{E} &\rightarrow \mathbb{C} \notag\\
		h &\rightarrow F(\varphi)(h) := F(h) \; .
	\end{align*}
\end{definition}
So evaluating a field $F(\varphi)$ as a functional amounts to replacing $\varphi$ by $h$ everywhere in the expression for $F$. A fundamental question in the construction is to determine the set of allowed fields $\mathscr{F}$ for the theory. In the definition of the Poisson bracket and the star product (chapter \ref{sec:defquant}) we will need expressions involving the pointwise products of distributions, e.g. of propagators and distributions appearing in the expression for a field $F$. Such products are -- if one requires associativity and the product law for derivatives of distributions to hold -- not a priori defined. However, there is a result by Hörmander characterizing the existence of such products in terms of the wave front sets of the involved distributions \cite[Thm. 8.2.10]{horm03}. The wave front set contains information about the singularities of a distribution, roughly speaking it describes the points where singularities are localized and the directions in Fourier space in which the distribution is singular at these points. Now Hörmander's criterion states that pointwise products of distributions whose wave front sets satisfy a certain condition can be defined meaningfully (see theorem \ref{thm:hoerm} in the appendix), and such a restriction (condition (ii) in the definition below)  is what we require from the allowed fields. Wave front sets and the corresponding theory of microlocal analysis are crucial to ensure that our construction works. However they are not needed for the proof of the MWI, so we will not dwell upon this topic and refer to Appendix A.3 where some of the main definitions and the criterion by Hörmander are given. The set of fields is now defined as follows.

\begin{definition}
	For $n \geq 1$ we define $\mathscr{F}'({\mathbb{M}^n})$ to be the space of all compactly supported distributions $f_n \in \mathscr{D}'(\mathbb{M}^n,\mathbb{C})$ whose 
	\begin{enumerate}[label=(\roman*)]
		\item integral kernels are symmetric in all their arguments: For all permutations $\pi$ in the symmetric group $S_n$, we have
		\begin{equation*}
			f_n(x_{\pi(1)},\dots,x_{\pi(n)}) = f_n(x_1,\dots,x_n) \; .
		\end{equation*}
		\item wave front sets satisfy the following property: 
		\begin{equation*}
			\text{WF}(f_n) \subseteq \left\{(x_1,\dots,x_n;k_1,\dots,k_n) \,| \,
			(k_1,\dots,k_n) \notin \overline{V}_+^{\times n}\cup
			\overline{V}_-^{\times n} \right\} \; .
		\end{equation*}
	\end{enumerate} 
	The \textit{set of fields} $\mathscr{F}$ is the space of all functionals $F\equiv F(\varphi): \mathcal{E} \rightarrow \mathbb{C}$ of the form 
	\begin{equation}
		\label{eq:fields}
		F(\varphi) = f_0 + \sum_{n=1}^{N} \int d x_1 \cdots d x_n \,
		\varphi(x_1) \cdots \varphi(x_n) \,
		f_n(x_1,\dots,x_n)
	\end{equation}
	where $N \in \mathbb{N}^*, f_0 \in \mathbb{C}$ and $f_n \in \mathscr{F}'({\mathbb{M}^n})$ for $n \geq 1$. 
\end{definition}

To require that the $f_n$ are compactly supported ensures that the integral in (\ref{eq:fields}) converges. Our point is that by choosing fields of this particular form we can grasp all physically relevant fields in a rigorous definition that still allows to perform explicit calculations. Note that the $F$ are off-shell fields: They are defined without reference to any field equation. Furthermore, the $F$ themselves are not distributions, since they are not linear in configurations $h$. They only involve objects $f_n$ which are distributions. We introduce two additional algebraic structures on $\mathscr{F}$ as follows. 

\begin{definition}
	The vector space $\mathscr{F}$ is endowed with a \textit{pointwise product}
	\begin{align*}
		\cdot \; : \mathscr{F} \times \mathscr{F} &\rightarrow \mathscr{F} \notag\\
		(F \cdot G) (h) &:= F(h) \cdot G(h)
	\end{align*}
	and a $*$\textit{-operation}
	\begin{align*}
		^*\, : \mathscr{F} &\rightarrow \mathscr{F} \notag\\ 
		F^* (h) &:= \sum_{n} \int d x_1 \cdots d x_n \,
		\varphi(x_1) \cdots \varphi(x_n) \,
		\overline{f_n(x_1,\dots,x_n)} \; .
	\end{align*}
This turns the space $(\mathscr{F}, \cdot, ^*)$ into a commutative, unital $*$-algebra called the algebra of classical fields. 
\end{definition}

To discuss spacetime symmetries we will need the notion of Poincaré covariance of fields, so we introduce the following action:

\begin{definition}
	A linear action $\beta$ of the proper, ortochronous Poincaré group $\mathscr{P}_+^{\uparrow}$ on $\mathscr{F}$ is defined by setting
	\begin{equation*}
		F \rightarrow \beta_{\Lambda,a}F := 
		\sum_{n=0}^{N} \int d x_1 \cdots d x_n \,
		\varphi(\Lambda x_1+a) \cdots \varphi(\Lambda x_n+a) \,
		f_n(x_1,\dots,x_n)
	\end{equation*}
	for $(\Lambda,a)\in\mathscr{P}_+^{\uparrow}$. 
\end{definition}

With our choice of allowed fields we can write down explicit expressions for functional derivatives.

\begin{definition}
	The $k$-th order \textit{functional derivative} of a Field $F\in\mathscr{F}$ with respect to the basic fields $\varphi(y_i)$ is defined by 
	\begin{align*}
		\frac{\delta^k F}{\delta \varphi(y_1) \cdots \delta \varphi(y_k)} := 
		\sum_{n=k}^{N} &\frac{n!}{(n-k)!} 
		\int d x_1 \cdots d x_{n-k} \, \notag\\
		&\varphi(x_1) \cdots \varphi(x_{n-k}) \,
		f_n(y_1,\dots,y_k,x_1,\dots,x_{n-k})\; .
	\end{align*}
\end{definition}
This functional derivative satisfies the Leibniz rule, and when applied to a configuration $h \in \mathcal{E}$ it yields a distribution in $\mathscr{D}'(\mathbb{M}^k,\mathbb{C})$ which is again compactly supported and symmetric in its arguments. We introduce the notion of the support of fields in a way that the term $f_0$ in (\ref{eq:fields}) does not contribute to the support of $F$.

\begin{definition}
	\label{def:supp}
	The support of $F\in\mathscr{F}$ is defined as 
	\begin{equation*}
		\text{supp}\,F := \text{supp}\,\frac{\delta F}{\delta \varphi(\cdot)}
	\end{equation*}
	where on the r.h.s. we mean the support of the distribution in $\mathscr{D}'(\mathbb{M},\mathbb{C})$. 
\end{definition}

A main idea in the construction of a relativistic field theory is that it should respect the principle of locality. One aspect is that the influence of phenomena taking place at a point propagate at a finite speed through spacetime -- so there is no action at a distance. Local fields are objects that can be expected to lead to theories respecting this principle. They are fields that may be written as the integral over a quantity that depends only on one point of spacetime. The definition reads: 

\begin{definition}
	The space $\mathscr{P}$ of field polynomials is the space of all polynomials in the variables $\{ \partial^a \varphi \,|\,a \in \mathbb{N}^d \}$ with real coefficients. The vector space $\mathscr{F}_{\text{loc}}\subset \mathscr{F}$ of \textit{local fields} is defined to be the set of all fields of the form 
	\begin{equation*}
		F = \sum_{i=1}^{K} \int dx\, A_i(x)\,g_i(x) \; , 
	\end{equation*}
	where $A_i \in \mathscr{P}, g_i \in \mathscr{D}(\mathbb{M})$ and $K \in \mathbb{N}^*$. We write $A(g) = \int dx \, A(x) \,g(x)$. 
\end{definition}

\begin{proposition}
	The functional derivative for the integral kernel of a local field monomial $A \in \mathscr{P}$ takes the form 
	\begin{equation}
		\label{eq:funcder}
		\frac{\delta A(x)}{\delta \varphi(y)} = 
		\sum_{a \in \mathbb{N}^d} (\partial^a \delta) (x-y)\frac{\partial A}{\partial (\partial^a \varphi)}(x) \; .
	\end{equation}
\end{proposition}

\subsection{Dynamics} 

Dynamics is governed by an action $S$, the field equations are obtained by variation of $S$ w.r.t. $\varphi$. We give the standard definition for the free scalar action. 

\begin{definition}
	The free action for scalar theory is the formal expression
	\begin{equation}
		\label{eq:S0}
		S_0 := \int dx\, \frac{1}{2} \left(\partial^{\mu} \varphi \, \partial_{\mu} \varphi - m^2 \varphi^2 \right) \; .
	\end{equation}
\end{definition}

However, $S_0 \notin \mathscr{F}$, since when written in the form of equation (\ref{eq:fields}), the corresponding $f_2$ is not compactly supported. Furthermore $S_0$ in general diverges when evaluated on a $h \in \mathcal{E}$. Restricting to only compactly supported field configurations $h$ has the consequence that no non-trivial solutions of the field equations exist \cite[chap. 1.5]{duetsch19}. Since in the following $S_0$ appears only as an index of other objects, we will stick with this formal definition and turn to the field equations.\footnote{In chapters \ref{sec:symms} and \ref{sec:complex} when discussing Noether's theorem in classical field theory calculations invol-ving $S_0$ are performed. A way to do these calculations rigorously would be to introduce the notion of a generalized Lagrangian and a corresponding action, which makes the notion precise of integrating out (\ref{eq:S0}) with compactly supported test functions \cite[chap. 4.1]{rejzner16}. However, for our proof of the MWI $S_0$ is only needed as an index.}

\begin{definition}
	Formal variation of the free action yields the free Klein-Gordon field equation
	\begin{equation*}
		\frac{\delta S_0}{\delta \varphi(x)} := -(\dalemb + m^2) \varphi(x) \; ,
	\end{equation*}
	which is a well defined expression.
\end{definition}

The interactions we consider will be local fields. Hence they are localized in space time by a test function $g$ and switched off outside the support of $g$. 

\begin{definition}
	The \textit{interactions} of the theory are of the form
	\begin{equation*}
		S = -\kappa \int dx \, g(x)\, L_{\text{int}}(x)
	\end{equation*}
	with coupling constant $\kappa \in \mathbb{R}$, $g \in \mathscr{D}(\mathbb{M})$ and interaction Lagrangian $L_{\text{int}} \in \mathscr{P}$,
	so $S \in \mathscr{F}_{\text{loc}}$. To make the dependence on $\kappa$ explicit we introduce the notation $\kappa \tilde{S} := S$. The \textit{total Lagrangian} is
	\begin{equation*}
		L_{\text{\normalfont{tot}}}(x) =  L_0(x)-\,\kappa\, g(x)L_{\text{\normalfont{int}}}(x) \;.
	\end{equation*}
	The \textit{field equation} is given by 
	\begin{equation*}
		\frac{\delta (S_0 + S)}{\delta \varphi(x)} = 0 \; .
	\end{equation*}
\end{definition}

Now we turn to the notion of on-shell fields. These are obtained by restricting the domain of definition of a given $F$ to solutions of the field equation. 

\begin{definition}
	The \textit{solution} space of the field equation is denoted by 
	\begin{equation*}
		\mathcal{E}_{S_0+S} := \left\{
		h \in \mathcal{E} \; \Big| \; 
		\frac{\delta (S_0 + S)}{\delta \varphi(x)} (h) = 0 \quad 
		\forall x \in \mathbb{M}
		\right\} \subseteq \mathcal{E} \; .
	\end{equation*}
	An \textit{interacting field} $F_S$ corresponding to the field $F \in \mathscr{F}_{\text{loc}}$ and an interaction $S$ is given by
	\begin{equation*}
		F_S = \left. F \right\vert_{\mathcal{E}_{S_0+S}} \; \qquad \text{or} \qquad A_S(x) = \left. A(x) \right\vert_{\mathcal{E}_{S_0+S}} \,. 
	\end{equation*}
	We call $F$ the \textit{off-shell} and $F_S$ the \textit{on-shell field}. 
\end{definition}

To clarify the transition from the classical to a quantum theory, the notion of a Poisson bracket will be important. We introduce such a bracket for the free action. Therefore some propagators are needed. 

\begin{definition}
	The \textit{retarded propagator} $\Delta_m^{\text{ret}} \in \mathscr{D}'(\mathbb{R}^d)$ for the free Klein-Gordon field equation is the (unique) fundamental solution to the equation
	\begin{equation*}
		(\dalemb + m^2) \Delta_m^{\text{ret}} = - \delta
	\end{equation*}
	satisfying $\text{supp}\,\Delta_m^{\text{ret}} \subseteq \overline{V}_+$. The corresponding \textit{commutator function} is 
	\begin{equation*}
		\Delta_m (x) = \Delta_m^{\text{ret}} (x) - \Delta_m^{\text{ret}} (-x) \; .
	\end{equation*}
\end{definition}

\begin{definition}
	The \textit{Poisson bracket} for the free action is a map $\{\cdot,\cdot\}_m: \mathscr{F} \times \mathscr{F} \rightarrow \mathscr{F}$ defined by 
	\begin{equation*}
		\{F,G\}_m = \int dx\,dy\, \frac{\delta F}{\delta \varphi(x)} \Delta_m(x-y) \frac{\delta G}{\delta \varphi(y)} \; . 
	\end{equation*}
\end{definition}

Taking the bracket of two local fields yields in general a non local expression. The Poisson bracket exists due to the wave front properties of the fields we required in (\ref{eq:fields}). It is a bilinear and skew-symmetric map satisfying the Leibniz rule and the Jacobi identity. Hence the space $\mathfrak{A}_{\text{cl}}:= (\mathscr{F}, \{\cdot,\cdot\}_m, \cdot, ^*)$ has the structure of a Poisson $*$-algebra, called the algebra of free classical fields. The subscript $m$ emphasizes that the Poisson bracket contains information about the field equation (or the action) via the commutator function $\Delta_m$ (we will mostly drop the subscript in the following). The field space $\mathscr{F}$ contains only kinematical information, dynamics is encoded in the algebraic structure of the Poisson bracket. 

\begin{proposition}
	\label{prop:retprop}
	The commutator function satisfies the following time zero relations 
	\begin{equation*}
	\Delta_m(0,\vec{x}) = 0\;,  \qquad
	\partial^0 \Delta_m(0,\vec{x}) = -\delta(\vec{x}) \; .
	\end{equation*}
	Using this, one gets the equal time Poisson bracket. 
\end{proposition}

\subsection{Perturbation Theory} 

Quantum field theory will be formulated in terms of perturbation theory. Here we introduce this formalism for the classical theory. The idea of perturbation theory is to express solutions of the full field equations as a power series in the coupling constant $\kappa$ which is taken to be small, where the lowest order of the series is the solution of the free field equation. To do so it makes sense to consider fields that propagate freely, then get perturbed in a localized region where interaction takes place and propagate further as perturbed fields. The notion of a retarded wave operator makes this idea precise. 

\begin{definition}
	A \textit{retarded wave operator} for a local interaction $S$ is a map $r_{S_0+S,S_0}: \mathcal{E} \rightarrow \mathcal{E}$ satisfying 
	\begin{enumerate}[label=(\roman*)]
		\item $r_{S_0+S,S_0} (h)(x) = h(x)$ for $x^0$ ``sufficiently early". 
		\item $\frac{\delta (S_0 + S)}{\delta \varphi(x)} \circ 
			r_{S_0+S,S_0} = \frac{\delta (S_0)}{\delta \varphi(x)}$ .
	\end{enumerate} 
	We will assume that a unique such operator is given for every interaction $S$. The \textit{retarded field} $F_S^{\text{ret}}$ corresponding to the local field $F$ and the interaction $S$ is defined by 
	\begin{equation*}
		F_S^{\text{ret}} := F \circ r_{S_0+S,S_0} : \mathcal{E} \rightarrow \mathbb{C} \; ,
	\end{equation*}
	and similarly for the integral kernel $A(x)$. 
\end{definition}

Note that the retarded wave operator and hence the retarded fields are defined on all $\mathscr{E}$ and not only for solutions of the free field equation. We now expand the retarded fields in terms of the coupling $\kappa$. 

\begin{definition}
	We define the \textit{classical retarded product} as a sequence of linear maps 
	\begin{equation*}
		R_{n,1} \equiv R_{\text{cl}} := \mathscr{F}_{\text{loc}}^{\otimes n} \otimes \mathscr{F}_{\text{loc}} \rightarrow \mathscr{F}\;, \qquad n \in \mathbb{N}
	\end{equation*}
	that are symmetric in the first $n$ entries, given by 
	\begin{equation*}
		R_{n,1}(\tilde{S}^{\otimes n},F) := \left. 
		\frac{d^n}{d \kappa^n} \right\vert_{\kappa = 0} F_{\kappa \tilde{S}}^{\text{ret}} \; .
	\end{equation*}
	We write this perturbative expansion symbolically as
	\begin{equation*}
		F_{\kappa \tilde{S}}^{\text{ret}} \simeq 
		\sum_{n = 0}^{\infty} \frac{\kappa^n}{n!} R_{n,1}(\tilde{S}^{\otimes n},F) \equiv R(e_{\otimes}^S,F) \; .
	\end{equation*}
\end{definition}

The term ``symbolic" here means that $R(e_{\otimes}^S,F)$ is considered to be an element of the space $\mathscr{F}[\![\hbar]\!]$ of formal power series in $\hbar$ with coefficients in $\mathscr{F}$, so no convergence of the series is implied (see definition \ref{def:formal_po}). The following proposition will be relevant for the discussion of the relation between classical and quantum fields. 

\begin{proposition}
	The (integral kernels of) two classical perturbative retarded fields $A, B\in\mathscr{P}$ factorize in the following sense \normalfont
	\begin{equation}
		\label{eq:clfact}
		(AB)_S^{\text{ret}}(x) = A_S^{\text{ret}} (x) \cdot B_S^{\text{ret}} (x) \;,
	\end{equation}
	\textit{as distributions in $x$, that is the pointwise product is well defined and commutative. They furthermore satisfy the off-shell field equation}
	\begin{equation}
		\label{eq:clpertoff}
		R_{\text{cl}} \left( e^S_{\otimes} , \frac{\delta (S_0+S)}{\delta \varphi (x)} 
		\right) 
		= \frac{\delta S_0}{\delta \varphi (x)} \; ,	
	\end{equation}
	\textit{which follows from the properties of the retarded wave operator.} 
\end{proposition}

	\section{Free Quantum Fields: Deformation Quantization}
	\label{sec:defquant}
	Deformation quantization is a  prescription of how to obtain a quantum theory from a given classical one. We briefly describe the main ideas and give the definitions needed later on. 

\subsection{The framework of deformation quantization}

Deformation quantization makes precise the idea that when passing from the classical into the quantum world, one should replace Poisson brackets by commutators and to get back the limit $\hbar \rightarrow 0$ has to be taken. In the framework of deformation quantization, this relation between theories is achieved by replacing the classical product  $\cdot$  with a non-commutative product $\star$ called ``star product" according to the following definition.

\begin{definition}
	A \textit{deformation quantization} for a Poisson algebra $(\mathscr{A},\cdot, \{\cdot,\cdot\})$ is a bilinear an associative product $\star$ on $\mathscr{A} \times \mathscr{A}$ with values in $\mathscr{A}[\![ \hbar ]\!]$ that satisfies 
	\begin{enumerate}[label=(\roman*)]
		\item $f \star g = f\cdot g + \mathcal{O} (\hbar)$, 
		\item $[f,g]_{\star} := f \star g - g \star f = i \hbar \{f,g\} + \mathcal{O} (\hbar^2)$,
	\end{enumerate} 
	for all $f,g\in\mathscr{A}$. 
\end{definition}

Condition (i) ensures that the $\hbar \rightarrow 0$ limit of the $\star$-product yields the classical product. Since we have no notion of convergence on the algebra of formal power series, taking this limit amounts to just setting $\hbar = 0$. Condition (ii) states that -- to lowest order in $\hbar$ -- the $\star$-commutator of the quantum theory is equal to the classical Poisson bracket. 

\begin{remark}
	We point out that the star product $\star$ and the star operation $*$ are different objects, although named confusingly similar.  
\end{remark}

\subsection{The star product}

We will now write down a particular $\star$-product which satisfies the requirements of the definition. Therefore, another kind of distribution is needed.

\begin{definition}
	\label{def:starprod}
	The \textit{Wightman two point function} is defined as the distribution with integral kernel 
	\begin{equation*}
		\Delta^+_m(x) := \frac{1}{(2\pi)^{d-1}} \int d^{d-1}\vec{p}\; 
		\frac{e^{-i(\omega_{\vec{p}} x^0-\vec{p}\vec{x})}}{2\omega_{\vec{p}}}\; , \quad \text{where} \quad \omega_{\vec{p}} := \sqrt{\vec{p}^2+m^2} \; .
	\end{equation*}
	It provides a splitting into positive and negative frequency parts of the commutator function: 
	\begin{equation*}
		i \Delta_m(x) = \Delta_m^+(x) - \Delta_m^+(-x) \; .
	\end{equation*}	
\end{definition}

\begin{definition}
	We define a \textit{star product} $\star : \mathscr{F}[\![ \hbar ]\!] \times \mathscr{F}[\![ \hbar ]\!] \rightarrow \mathscr{F}[\![ \hbar ]\!]$ as
	\begin{align}
		\label{eq:star}
		F \star G = 
		\sum_{n=0}^{\infty}\frac{\hbar^n}{n!}
		\int dx_1 \cdots dx_n \, dy_1 \cdots dy_n \notag\\
		\cdot\; 
		\frac{\delta^n F}{\delta \varphi(x_1)\cdots\delta\varphi(x_n)} 
		\;\prod_{l=1}^{n} \Delta_m^+(x_l - y_l)\;
		\frac{\delta^n G}{\delta \varphi(y_1)\cdots\delta\varphi(y_n)} \; . 
	\end{align}
	The corresponding $\star$-commutator is 
	\begin{equation*}
			[F,G]_{\star} := F \star G - G \star F \; .
	\end{equation*}
\end{definition}

Equation (\ref{eq:star}) amounts to the prescription that when computing the star product of two local fields $F$ and $G$, one has to find all possible contractions between the basic fields contained in $F$ and those in $G$. Hence it encodes the combinatorics of Wick's theorem. If in QFT the fields are represented as operators on some Hilbert space, the star product corresponds to the operator product.

\begin{theorem}[existence of the star product]
	\label{thm:starprod}
	\normalfont{\cite[chap. 2.4]{duetsch19}}
	\textit{The star product exists (due to the wave front properties of the fields and the two point function), and is a star product in the sense of Definition} \normalfont{\ref{def:starprod}}\textit{. The corresponding commutator satisfies the Jacobi Identity and acts as a derivation on the algebra of fields.}
\end{theorem}

The Algebra $\mathfrak{A} = ( \mathscr{F}[\![ \hbar ]\!],\star,^*)$ is called the free algebra of quantum fields. One property of the $\star$-commutator is that it yields zero for fields whose supports lay at spacelike distances of each other. This can be interpreted in relation to what we called above the principle of locality: Phenomena taking place at spacelike distances can't have any influence on each other. The actual statement reads as follows:

\begin{proposition}[spacelike commutativity]
	\label{prop:spacecomm}
	Let $F,G\in \mathscr{F}$. If $(x-y)^2 < 0$ for all \normalfont{$x \in \text{supp}\, F$}\textit{ and }\normalfont{$y \in \text{supp}\, G$}\textit{, then $[F,G]_{\star} = 0$.}
\end{proposition}

	\section{Interactions: Perturbative QFT}
	\label{sec:interact}
	In section $\ref{sec:clft}$ we have presented a perturbative formulation of classical field theory where the main objects are retarded products $R_{n,1}$. Now we turn to the quantum theory, which we will describe in terms of time ordered products $T_n$. A formulation of QFT using retarded products is also possible, and both retarded and time ordered products yield equivalent descriptions.

\subsection{The time ordered product}

While the retarded products $R_{n,1}$ are the expansion coefficients of the retarded fields, the time ordered products  $T_n$ are the expansion coefficients of the $S$-matrix. The physical relevance of the $S$-matrix lays in its role for the description of scattering experiments. In common text book approaches it is taken to be an operator mapping in-states to out-states that can be computed via the Dyson series \cite[chap. 4.2]{peskin95} as (with $\hbar=1$) 
\begin{equation*}
	\textbf{S}(L) = T \left\{\exp \Big[
	i \int dx \, L_{\text{int}}(x)
	\Big]
	\right\} \; ,
\end{equation*}
where $T$ denotes time ordering, that is all terms in the bracket containing an $x$ should be rearranged by putting ``later times to the left". However, this time ordering operation is not well defined for the case of distributions. \\

Our approach is a different one. To construct the $T$-products we use the framework of causal perturbation theory going back to Epstein and Glaser \cite{EG73}. We define the $T$-products axiomatically, where the axioms are motivated by properties that hold true in the classical theory, and then show that we can construct such objects. The axioms are divided into two classes, of which the first one are the basic axioms. In the inductive construction, they determine the $T$-products uniquely up to points where all their arguments are equal. The possible extensions to such points are not unique, but we restrict them by requiring further properties, called renormalization conditions, which form the second class. 

\subsection{Axioms for the time ordered product}

In the following we give the four basic axioms for the $T$-product. 

\begin{definition}
	\label{def:Tn}
	For $n\in\mathbb{N}$ we define the \textit{time-ordered product} of $n$-th order as a map 
	\begin{equation*}
		T_{n} : \mathscr{F}_{\text{loc}}^{\otimes n} 
		\rightarrow \mathscr{F}
	\end{equation*} 
	satisfying the following axioms: 
	\begin{enumerate}[label=(\roman*)]
		\item \textbf{Linearity:} $T_n$ is a linear map. 
		\item \textbf{Initial condition:} $T_1(F)=F$ for all $F \in \mathscr{F}_{\text{loc}}$. 
		\item \textbf{Symmetry:} $T_n$ is a totally symmetric map
		\begin{equation*}
			T_n\left(F_{\pi(1)},\dots,F_{\pi(n)}\right) = T_n\left(F_1,\dots,F_n\right) \quad \forall \pi \in S_n,\; F_1, \dots, F_n \in \mathscr{F}_{\text{loc}} \;.
		\end{equation*} 
		\item \textbf{Causality:} For any $A_1,\dots,A_n \in \mathscr{P}$, $T_n$ factorizes causally. That is 
		\begin{align}
			\label{eq:causality}
			T_n\left(A_1(x_1),\dots,A_n(x_n)\right) = \; &T_k\left(A_{\pi(1)}(x_{\pi(1)}),\dots,A_{\pi(k)}(x_{\pi(k)})\right) \notag\\ 
			&\star 
			T_{n-k}\left(A_{\pi(k+1)}(x_{\pi(k+1)}),\dots,A_{\pi(n)}(x_{\pi(n)})\right)
		\end{align}
		whenever $\{x_{\pi(1)},\dots,x_{\pi(k)}\} \cap 
		\left(\{x_{\pi(k+1)},\dots,x_{\pi(n)}\}+ \overline{V}_-\right) = \emptyset$ for a permutation $\pi \in S_n$. 
	\end{enumerate}
\end{definition}

\begin{definition}
	The \textit{S-matrix} is defined as 
	\begin{equation*}
		\textbf{S}(F) = \sum_{n=0}^{\infty}
		\frac{i^n}{n! \hbar^n}\, T_n\left(F^{\otimes n}\right)
		\equiv T( e_{\otimes}^{iF/\hbar}) \; .
	\end{equation*}
\end{definition}

Axiom (iv) is the translation of the ``time ordering" prescription described above into our framework. In the perturbative setting it is equivalent to the following property of the $S$-matrix:
\begin{equation*}
	\textbf{S} (H+F) = \textbf{S} (H)\star \textbf{S} (F)\quad \text{whenever} \quad \text{supp}\, H \cap {\text{supp}F\,+ \overline{V}_-} = \emptyset
\end{equation*}
If we interpret the $S$-matrix as a scattering operator the physical idea is the following: Whenever the interaction with $H$ does not lay in the past of the interaction with $F$, then the scattering with $H$ and $F$ can be described as two separated scattering processes taking place one after another. 

\begin{remark}
	The fact that the time-ordered products depend only on local functionals implies the Action Ward Identity
	\begin{equation}
		\label{eq:AWI}
		\partial_{x_l}T_n\left(\cdots \otimes A(x_l) \otimes \cdots \right) 
		= T_n \left( \cdots \otimes \partial_{x_l} A(x_l) \otimes \cdots \right) \quad \forall A \in \mathscr{P}, 1\leq l\leq n\; .
	\end{equation}
\end{remark}

\subsection{Inductive construction of the time ordered products}

We now want to construct maps $T_n$ satisfying the above axioms. This is done by induction on $n$, where axiom (ii) provides the basis of the induction. Define the thin diagonal as $\Delta_n := \{(x_1,\dots,x_n)\in\mathbb{M}^n\,|\,x_1=\cdots=x_n\}$. The idea is to work with an open cover of $\mathbb{M}^n\backslash\Delta_n$, where on each set of the cover the $T$-product factors causally and is uniquely determined by the products of lower orders through axiom (iv). This leads to the following result: 

\begin{theorem}
	\normalfont{\cite[chap. 3.3.2]{duetsch19}}
	\textit{Given $T_1,\dots,T_n$, the basic axioms determine $T_{n+1}$ uniquely on the space $\mathscr{D}(\mathbb{M}^n\backslash \Delta_n)$.}
\end{theorem}

On the smaller subspace  $\check{\mathbb{M}}^n := \{(x_1,\dots,x_n) \in \mathbb{M}^n \; | \; x_l \neq x_j \; \forall\, 1 \leq l \leq j \leq n\}$ there is a way to compute the $T_n$ explicitly via the Feynman propagator (that exists as a distribution, how can be shown using Hörmander's criterion):  

\begin{definition}
	The \textit{Feynman propagator} is the distribution whose integral kernel is defined by
	\begin{equation*}
		\Delta_m^F(x) = \theta (x^0) \Delta_m^+(x)
		+ \theta (-x^0) \Delta_m^+(-x) \; .
	\end{equation*}
\end{definition}

\begin{theorem}
	On $\mathscr{D}(\check{\mathbb{M}}^n)$ it holds that 
	\begin{equation}
		\label{eq:unrenT}
		T_n \left( A_1(x_1),\dots,A_n(x_n) \right) = 
		A_1(x_1) \star_{F} \cdots  \star_{F}  A_n(x_n)
	\end{equation}
	for all $A_1,\dots,A_n\in\mathscr{P}$. On the r.h.s, the star product $\star_{F}$ is obtained by replacing the two-point function $\Delta^+$ with the Feynman propagator $\Delta^F$ in the definition of the star product. This expression is called the the \normalfont{unrenormalized $T$-product}. 
\end{theorem}

The next task is to extend the time ordered products to the thin diagonal $\Delta^n$. This step is called renormalization, it is no more unique. We require properties from the renormalized $T$-products that restrict the possible extensions. To formulate them we will need some further notions, which we introduce in the following three definitions. 

\begin{definition}
	A \text{state} $\omega$ on a unital $*$-algebra $\mathscr{A}$ is linear a map $\omega :  \mathscr{A} \rightarrow \mathbb{C}$ which is 
	\begin{enumerate}[label=(\roman*)]
		\item real: $\omega(F^*) = \overline{\omega(F)}$,
		\item positive: $\omega(A^* A) \geq 0$,
		\item normalized: $\omega(1)=1$. 
	\end{enumerate} 
	We define the \textit{vacuum state} $\omega_0$ on the algebra of quantum fields as 
	\begin{align*}
		\omega_0 : \mathfrak{A} &\rightarrow \mathbb{C} \notag\\
		F &\mapsto f_0 \; ,
	\end{align*}
	where $f_0$ as in equation (\ref{eq:fields}). Lowercase letters will be used to denote the vacuum expectation values (VEVs) of objects, e.g. $t_n(A_1,\dots,A_n)=\omega_0\left(T_n\left(A_1,\dots,A_n\right)\right)$. 
\end{definition}

\begin{definition}
	For a monomial $\partial^a \varphi \in \mathscr{P}$ where $a\in\mathbb{N}^d$, its \textit{mass dimension} is defined as
	\begin{equation*}
	\text{dim}\,\partial^a \varphi := \frac{d-2}{2}
	+|a| \; .
	\end{equation*}
	Let $\mathscr{P}_j$ be the vector space spanned by all monomials $A\in \mathscr{P}$ with $\text{dim}\,A=j$. We define the set of homogeneous polynomials (w.r.t their mass dimension) as the union $\mathscr{P}_{\text{hom}}:=\bigcup\limits_{j\in\mathbb{N}}\mathscr{P}_j$.
\end{definition} 

\begin{definition}
	The \textit{field parity transformation} on $\mathscr{F}$ corresponding to the mapping $\varphi \mapsto - \varphi$ is defined as 
	\begin{align*}
		\alpha : \mathscr{F} &\rightarrow \mathscr{F} \notag\\
		(\alpha F)(h) &= F(-h) \quad \forall h \in \mathcal{E} \; .
	\end{align*}
\end{definition}

With these additional definitions, the conditions that we impose when extending the $T$-products to the thin diagonal may be formulated. They are motivated by properties of the classical theory that we want to maintain as far as possible in the quantum theory. 

\begin{definition}
	\label{def:rencon}
	The \textit{renormalization conditions} for the $T$-product are
	\begin{enumerate}[label=(\roman*)]
		\setcounter{enumi}{4}
		\item \textbf{Field independence:} 
		\begin{equation*}
			\frac{\delta T_n(F^{\otimes n})}{\delta \varphi(x)}
			= n T_n \left( \frac{\delta F}{\delta \varphi(x)}
			\otimes F^{\otimes(n-1)} \right)
		\end{equation*}
		\item \textbf{$*$-structure and field parity:}
		\begin{equation*}
			\textbf{S}(F)^* = \textbf{S}(F^*)^{\star -1} \;\; \forall F \in \mathscr{F}_{\text{loc}}
			\qquad \text{and} \qquad 
			\alpha \circ T_n = T_n \circ \alpha^{\otimes n} \; .
		\end{equation*}
		\item \textbf{Poincaré Covariance:}
		\begin{equation*}
			\beta_{\Lambda,a} \circ T_n = T_n \circ \beta_{\Lambda,a}^{\otimes n} \quad \forall (\Lambda, a) \in \mathscr{P}^{\uparrow}_+
		\end{equation*}
		\item \textbf{Further symmetries:} If the unrenormalized $T_n$ satisfy additional symmetries, we require them to hold also for the renormalized $T_n$ (for more details see chapter \ref{sec:qMWI}). 
		\item \textbf{Off-shell field equation:}
		\begin{align*}
			T_n \left(\varphi(g) \otimes F_1 \cdots \otimes F_{n-1}\right) &= \varphi(g)\, T_{n-1} \left( F_1 \otimes\cdots F_{n-1} \right) \notag\\
			+\hbar \int dx \,dy\,g(x) &\Delta_m^F(x-y) \frac{\delta}{\delta\varphi(x)}
			T_{n-1} \left( F_1 \otimes\cdots \otimes F_{n-1} \right)
		\end{align*}
		\item \textbf{$\hbar$-dependence:}
		\begin{equation*}
			t(A_1\dots,A_n) \sim \hbar^{\sum_{j=1}^{n} |A_j|/2}
		\end{equation*}
		for all monomials $A_1,\dots,A_n$ which fulfill $A_j \sim \hbar^0 \; \forall j$. The order of a monomial $A=c\prod_{l=1}^{L}\partial^{a_l}\varphi$ in $\varphi$ is defined as $|A|:=L$. 
	\end{enumerate}
\end{definition}

That these conditions are really renormalization conditions is not obvious, but can be shown to hold true. 

\begin{proposition}
	The unrenormalized $T$-products satisfy all renormalization conditions. 
\end{proposition}

If the renormalization conditions are satisfied they imply the following statements. 
\begin{theorem}
	\label{thm:cWick}
	The field independence axiom \normalfont{(v)} \textit{implies the validity of the causal Wick expansion. Let $A_1,\dots,A_n \in \mathscr{P}$ be monomials and write}
	\begin{align*}
		\underline{A} & := \frac{\partial^k A}{\partial(\partial^{a_1}\varphi)\cdots\partial(\partial^{a_k}\varphi)} \neq 0 \; , \notag\\
		\overline{A} & := C_{a_1\dots a_n} \partial^{a_1}\varphi\cdots\partial^{a_n}\varphi \;,
	\end{align*}
	\textit{where $C_{a_1\dots a_n}$ is a combinatorial factor. The causal Wick expansion for a time-ordered product reads}
	\begin{align}
		\label{eq:cwick}
		T_n&\left( A_1(x_1)\otimes\cdots\otimes A_n(x_n) \right)\notag\\
		&= \sum_{\underline{A}_l\subseteq A_l}\omega_0 \left(T_n\left( \underline{A}_1(x_1),\dots,\underline{A}_n(x_n)
		\right)\right)
		\overline{A}_1(x_1)\cdots\overline{A}_n(x_n) \; ,
	\end{align}
	\textit{where the sum runs over all $k\in\mathbb{N}$ and $a_1,\dots a_k \in \mathbb{N}^d$ which yield a non-vanishing $\underline{A}$. }
\end{theorem}

\begin{proposition}
	Furthermore field independence implies that the kernels of the $T$-products are localized at their arguments, that is
	\begin{equation}
		\label{eq:tsupp}
		\text{\normalfont{supp}}\, T\left(A_1(x_1),\dots,A_n(x_n)\right)
		\subseteq \{x_1,\dots,x_n\} \; ,
	\end{equation}
	where definition \normalfont{\ref{def:supp}} \textit{of the support is used.}
\end{proposition}

\begin{proposition}
	\label{prop:t_transl}
	From Poincaré covariance and the fact that $\omega_0 \circ\beta_{\Lambda,a} = \omega_0$ it follows that the VEVs of the $T$-products depend only on their relative coordinates 
	\begin{equation*}
		\omega_0\big(\,T_n\left( A_1(x_1),\dots A_n(x_n) \right)\big) = 
		t_n\left(A_1,\dots,A_n\right)(x_1-x_n,\dots x_{n-1}-x_n) \; ,
	\end{equation*}
	that is they are translation invariant, numerical ($\mathbb{C}$-valued) distributions. 
\end{proposition}

The open question is now if there exist extended $T_n$ satisfying all the conditions and if so, how much arbitrariness is left in their choice. From theorem \ref{thm:cWick} and proposition \ref{prop:t_transl} we see that we can express any unrenormalized $T$-product as the sum over translation invariant numerical distributions. So to extend the $T$-products it is sufficient to extend all their VEVs $t$ to the origin. The uniqueness of this procedure may be characterized in terms of the scaling degree of a distribution, which -- roughly speaking -- gives a measure of the strength of its singularity at $x=0$. 

\begin{definition}
	The \textit{scaling degree} (w.r.t the origin) of a distribution $t \in \mathscr{D}'(\mathbb{R}^k)$ or $t \in \mathscr{D}'(\mathbb{R}^k\backslash\{0\})$ is given by
	\begin{equation*}
	\text{sd}(t) = \text{inf}\, \{ r \in \mathbb{R} \; | \lim\limits_{\lambda \searrow 0} \lambda^r \, t(\lambda x) = 0  \} \; .
	\end{equation*}
	We set $\text{inf}\,\emptyset:= \infty$ and $\text{inf}\,\mathbb{R}:= -\infty$. 
\end{definition}

The possible extensions of distributions to the origin are characterized by the following theorem, due to \cite[chap. 5]{fred00}. 

\begin{definition}
	An \textit{extension} of a distribution $r^0 \in \mathscr{D}'(\mathbb{R}^k\backslash\{0\})$ is a distribution $r \in \mathscr{D}'(\mathbb{R}^k)$ such that $r(f)=r^0 (f) \;\; \forall f \in \mathscr{D}(\mathbb{R}^k\backslash\{0\})$. 
\end{definition}

\begin{theorem}[extensions of distribution] \textit{
	Let $t^0 \in \mathscr{D}'(\mathbb{R}^k\backslash\{0\})$.}
	\begin{enumerate}[label=(\roman*)]
		\normalfont
		\item
		\textit{If $\text{sd}\,(t^0)<k$, there exists a unique extension $t \in \mathscr{D}'(\mathbb{R}^k)$ fulfilling the condition $\text{sd}\,(t)=\text{sd}\,(t^0)$.}
		\normalfont
		\item 
		\textit{If $k\leq\text{sd}\,(t^0)<\infty$, several extensions $t \in \mathscr{D}'(\mathbb{R}^k)$ satisfying $\text{sd}\,(t)=\text{sd}\,(t^0)$ exist. The difference of two such solutions $t$ and $t'$ is of the form}
		\begin{equation}
		\label{eq_fredh}
		t'-t = \sum_{|a|\leq\text{sd}(t^0)-k} C_a \,\partial^a \delta_{(k)} \qquad \text{\textit{where}} \quad C_a \in \mathbb{C} \; .
		\end{equation} 
	\end{enumerate}
\end{theorem}

\begin{definition}
	\label{def:finren}
	In case (ii), the addition of a term $\sum_{a} C_a\, \partial^a \delta_{(k)}$ is called a \textit{finite renormalization}. 
\end{definition}

So the freedom of renormalization consists in choosing the constants $C_a$ accordingly to equation (\ref{eq_fredh}). This choice is what is being restricted by the renormalization conditions. We require one additional condition, concerning the scaling degree as follows:

\begin{enumerate}[label=(\roman*)]
	\setcounter{enumi}{10}
	\item \textbf{Scaling degree:}
	\begin{equation*}
	\text{sd}\,t(A_1,\dots,A_n)(x_1-x_n,\dots) 
	\leq \sum_{j=1}^{n} \text{dim} \, A_j \quad 
	\forall A_1,\dots,A_n \in \mathscr{P}_{\text{hom}}
	\end{equation*}
\end{enumerate}

Knowing the specific form of the finite renormalizations one may prove the following theorem: 

\begin{theorem} 
	\normalfont{\cite[chap. 3.2.4]{duetsch19}}
	 \textit{There exists a sequence $T_n$ of maps defined on the whole $\mathscr{D}(\mathbb{M}^n)$ satisfying all axioms and all the renormalization conditions.}
\end{theorem}

This completes our construction of the $T$-products. \\

We have described the inductive construction for the time ordered products. The same can be done for the retarded products of quantum field theory, corresponding to the classical ones described in section \ref{sec:clft}. The particular axioms and renormalization conditions differ, but give an equivalent description of the quantum theory. Transitions between the two can be done by using Bogoliubov's formula
\begin{equation}
	\label{eq:bog}
	R(e_{\otimes}^{F/\hbar},G) 
	= \frac{\hbar}{i} \left.\frac{d}{d\lambda} \right\vert_{\lambda=0} 
	\textbf{S}(F)^{\star -1} \star \textbf{S}(F+\lambda G) \; .
\end{equation}
The constructions for $R$- and $T$-products are equivalent in the following sense. Assume the axioms on linearity and symmetry to hold for both the $R$ and $T$ products. Then constructing either of them satisfying the axioms determines the other uniquely and according to the respective axioms, by equation (\ref{eq:bog}).

	\chapter{Relating classical to quantum symmetries}
	\label{chap:symms}

	\section{Generalities about the MWI}
	\label{sec:symms}
	The formulation we use here to describe symmetries is the Master Ward Identity (MWI). It is a relation that holds in classical field theory and that we want to require from the quantum theory. 

\subsection{The classical MWI for a general action}

In this section we follow \cite[chap. 4]{duetsch19}. The classical MWI is the following relation. 

\begin{proposition}
	Let $Q\in\mathscr{P}$ and an interaction $S$ be given. From the perturbative off-shell field equation \textnormal{(\ref{eq:clfact})} and the classical factorization property \textnormal{(\ref{eq:clpertoff})} it follows that
	\begin{align}
		\label{eq:cMWI}
		R_{\text{\normalfont{cl}}} \left( e^S_{\otimes} , Q(x) \cdot \frac{\delta (S_0+S)}{\delta \varphi (x)}  
		\right)
		&= R_{\text{\normalfont{cl}}} \left(e^S_{\otimes}, Q(x)\right) \cdot R_{\text{\normalfont{cl}}} \left( e^S_{\otimes} , \frac{\delta (S_0+S)}{\delta \varphi (x)} 
		\right) \notag\\ 
		&= R_{\text{\normalfont{cl}}}\left( e^S_{\otimes} , Q(x)\right)
		\cdot \frac{\delta S_0}{\delta \varphi (x)} \; .
	\end{align}
	This is called the Master Ward Identity (MWI) for the retarded products. 
\end{proposition}

Both the off-shell field equation and the factorization hold true in classical theory. So the MWI is a general relation that follows from the properties of the fields and hence is always valid. We reformulate it in the following way.

\begin{definition}
	We define the functional
	\begin{equation}
		\label{eq:funcA}
		A = \int dx\, \sum_{k=1}^{K} h_i(x)\, Q_i(x)\, \frac{\delta S_0}{\delta \varphi(x)} \quad \text{where} \quad K \in \mathbb{N} ,\quad Q_i \in \mathscr{P} , \quad h_i \in  \mathscr{D}(\mathbb{M}) \; ,
	\end{equation}
	and a corresponding derivation 
	\begin{equation}
		\label{eq:derQ}
		\delta_{\vec{h}\cdot\vec{Q}} := 
		\int dx\, \sum_{k=1}^{K} h_i(x)\, Q_i(x)\, \frac{\delta}{\delta \varphi(x)} \; .
	\end{equation}
\end{definition}

\begin{proposition}
	The classical MWI for the symmetry $\vec{Q}$ and the interaction $S$ may be written as 
	\begin{equation}
		\label{eq:ret_MWI}
		R_{\text{\normalfont{cl}}} \left(e^S_{\otimes}, (A+\delta_{\vec{h}\cdot\vec{Q}} S)\right)
		= \int dx\, \sum_{k=1}^{K} h_i(x)\, R_{\text{\normalfont{cl}}} \left(e^S_{\otimes}, Q_i(x)\right)\cdot \frac{\delta S_0}{\delta \varphi(x)} \; .
	\end{equation}
\end{proposition}

Why we call $\vec{Q}$ a symmetry will become clear when discussing Noether's theorem in the next section. We can translate the MWI into the quantum theory by just replacing classical retarded products with the quantum ones, and reformulate it for the $T$-products by Bogoliubov's equation. This yields the following. 

\begin{proposition}
	The quantum MWI for the $T$-product and $K=1$ reads 
	\begin{equation*}
		T\left(e^{iS/\hbar}_{\otimes}\otimes (A+\delta_{hQ} S)\right)
		= \int dx\, h(x)\, T\left(e^{iS/\hbar}_{\otimes}\otimes Q(x)\right)\cdot \frac{\delta S_0}{\delta \varphi(x)} \; .
	\end{equation*}
	Writing it to $n$-th order for non-diagonal entries yields
	\begin{align}
		\label{eq:nthMWI}
		T_{n+1} &\left(F_1\otimes \cdots F_n \otimes A\right)
		+ \frac{\hbar}{i} \sum_{l=1}^{n} 
		T_{n} \left(F_1\otimes \cdots \otimes \delta_{hQ} F_l \otimes \cdots \otimes F_n\right) \notag\\
		&= \int dx\, h(x)\, 
		T_{n+1} \left(F_1\otimes \cdots F_n \otimes Q(x)\right) 
		\cdot \frac{\delta S_0}{\delta \varphi(x)} \; .
	\end{align}
\end{proposition}

The status of this equation -- under which conditions it holds true -- is not yet clarified. This will be discussed in section \ref{sec:qMWI}. 

\subsection{Relation to Noether's Theorem}

This section discusses the relation of the MWI to Noether's theorem in classical field theory. We start by defining smooth transformations of fields.\footnote{We do not give any meaning to the notion of ``smoothness" for this kind of transformations. What we want is that all derivatives w.r.t. $\alpha$ exist and satisfy the product rule.}

\begin{definition}
	\label{def:shiva}
	Let a smooth \textit{transformation} of the basic field depending on the parameter $\alpha$ be given by a mapping $\varphi \mapsto \varphi_{\alpha}$ which may involve expressions depending explicitly on $x$. We define a transformation of a general field $F\in\mathscr{F}$ by
	\begin{equation*}
	F(\varphi) \mapsto F_{\alpha}(\varphi) := F(\varphi_{\alpha}) \;.
	\end{equation*}
	The corresponding \textit{infinitesimal transformation} $s$ is 
	\begin{equation*}
	s F := \left.\frac{\partial}{\partial\alpha}\right\vert_{\alpha=0} F_{\alpha} \; .
	\end{equation*}
	We require the following properties to hold true: 
	\begin{align*}
		(F_1\cdot F_2)_{\alpha} &= (F_1)_{\alpha} \cdot (F_2)_{\alpha}
		\notag\\
		s(F_1\cdot F_2) &= (s F_1)\cdot F_2 + F_1 \cdot (s F_2) \notag\\
		s(\partial^{\mu} A(x)) &= \partial^{\mu}(s A(x)) \quad \text{\normalfont{for}}\; A \in \mathscr{P}
	\end{align*}
\end{definition}
With these definitions we can give a version of Noether's theorem for theories with one single basic field. 

\begin{theorem}
	Consider a system with a total action containing the field $\varphi$ and its first derivative and a transformation $\varphi \mapsto \varphi_{\alpha}$ leaving the total action invariant, that is $(S_{\text{\normalfont{tot}}})_{\alpha}= S_{\text{\normalfont{tot}}} \equiv S_0 +S$. Assume the test function $g$ in the total Lagrangian $L_{\text{\normalfont{tot}}}(x) =  L_0(x)-\,\kappa\, g(x)L_{\text{\normalfont{int}}}(x)$ to be constant on a neighborhood of a double cone $\mathscr{O}$. Then for $x \in \mathscr{O}$ there exists a four vector $j^{\mu} \in \mathscr{P}$ called the Noether current satisfying 
	\begin{equation}
	\label{eq:noeth}
	\partial_{\mu} j^{\mu}(x) = Q(x)\, \frac{\delta (S_0 +S)}{\delta \varphi(x)} \; ,
	\end{equation}
	where $Q := s \varphi$ for the infinitesimal transformation $s$ corresponding to $\alpha$. The current $j^{\mu}$ is given by
	\begin{equation}
	\label{eq:current}
	j^{\mu}(x) =  \Lambda^{\mu}(x) - \frac{\partial L_{\text{\normalfont{tot}}}}{\partial (\partial_{\mu}\varphi)} \, s \varphi
	\end{equation}
	for some $\Lambda^{\mu} \in \mathscr{P}$. 
\end{theorem}

\begin{proof}
	Since $S_{\text{tot}}$ is invariant under $\alpha$ we have $s S_{\text{tot}} =0$, that is $sL_{\text{\normalfont{tot}}} =  \partial_{\mu}\Lambda^{\mu}$ for some $\Lambda^{\mu} \in \mathscr{P}$. Using the derivation property of $s$ and the fact that $s$ commutes with derivatives (both from from definition \ref{def:shiva}) we calculate for $x\in\mathscr{O}$
	\begin{align*}
	\partial_{\mu}\Lambda^{\mu}(x) = s L_{\text{tot}}(x) 
	&= \frac{\partial L_{\text{tot}}}{\partial \varphi}(x)\, s \varphi(x) + 
	\frac{\partial L_{\text{tot}}}{\partial(\partial_{\mu} \varphi)}(x)\, s \big(\partial_{\mu}\varphi(x)\big) \notag\\
	&=  \partial_{\mu}^x \Big(\frac{\partial L_{\text{tot}}}{\partial(\partial_{\mu} \varphi)}(x)\, s \varphi(x) \Big)
	+ \Big[ \frac{\partial L_{\text{tot}}}{\partial \varphi}(x)
	- \partial_{\mu} \frac{\partial L_{\text{tot}}}{\partial(\partial_{\mu} \varphi)}(x)\Big]  s \varphi(x) \; .
	\end{align*}
	Since $g$ is constant on $\mathscr{O}$, there are no further contributions. Due to equation (\ref{eq:funcder}), the expression in the $[\cdots]$-brackets equals $\delta S_{\text{tot}} / \delta \varphi(x)$. So defining $Q$ and $j^{\mu}$ as in the proposition yields the result.  
\end{proof}

Noether's theorem tells us that -- for systems of the described kind -- the current in equation (\ref{eq:noeth}) vanishes if the field goes on-shell, that is if it satisfies the equation of motion. Now we see that the r.h.s. of (\ref{eq:noeth}) is precisely the argument of the retarded product on the l.h.s. of the MWI (\ref{eq:cMWI}). By putting this into the MWI and using the Action Ward Identity we get
\begin{align*}
	\partial_{\mu}^x R_{\text{\normalfont{cl}}} \left( e^S_{\otimes} , j^{\mu}(x) 
	\right)
	= R_{\text{\normalfont{cl}}}\left( e^S_{\otimes} , Q(x)\right)
	\cdot \frac{\delta S_0}{\delta \varphi (x)} \; , 
\end{align*}
where now $Q=s\varphi$ for the corresponding transformation. So in this sense, the MWI covers the description of symmetries that can be expressed as the conservation of a current via Noether's theorem. 

\subsection{The MWI in the quantum theory}
\label{sec:qMWI}

We will now clarify the status of the MWI in the quantum theory by the following proposition: 

\begin{proposition}
	The MWI is a renormalization condition for the quantum theory. 
\end{proposition}

That is the MWI is always satisfied for unrenormalized $T$-products, but it must be imposed as a condition to hold for the extensions to the thin diagonal. Hence the MWI does a priori not hold, and it is an open question for each individual model whether it can be satisfied by choosing the renormalizations appropriately or not. So the question arises about properties of the terms violating the quantum MWI. The following theorem due to \cite[chap. 5.2]{db08} describes their structure. We assume our interactions to be local of the form $S=\kappa L(g)$ with $g\in\mathscr{D}(\mathbb{M})$, $L \in \mathscr{P}$. 

\begin{theorem}[anomalous MWI]
	\label{thm:aMWI}
	
	Let $(T_n)_{n\in\mathbb{N}}$ be a time ordered product satisfying all basic axioms and the renormalization conditions translation covariance and field independence. Then there exists a unique sequence of linear maps \normalfont
	\begin{align*}
		\Delta^n : \mathscr{P}^{\otimes(n+1)} &\rightarrow \mathscr{D}'(\mathbb{M}^{n+1},\mathscr{F}_{\text{loc}}) \; ; 
		\notag\\
		\otimes^n_{j=1} L_j \otimes Q &\mapsto 
		\Delta^n\left(\otimes^n_{j=1} L_j(x_j) ;Q(y)
		\right)
		\equiv \Delta^n\left(( \otimes^n_{j=1} L_j) \otimes Q
		\right)
		(x_1,\dots,x_n,y)
	\end{align*}
	\textit{
	that are totally symmetric in the first $n$ entries and fulfill the anomalous MWI }
	\begin{equation}
		\label{eq:anoMWI}
		T\left(e^{iS/\hbar}_{\otimes}\otimes (A+\delta_{hQ} S+\Delta(e^S_{\otimes};hQ))\right)
		= \int dx\, h(x)\, T\left(e^{iS/\hbar}_{\otimes}\otimes Q(x)\right)\cdot \frac{\delta S_0}{\delta \varphi(x)} \; .
	\end{equation}
	\textit{
	The maps $\Delta^n$ have the following properties: }
	\begin{enumerate}[label=(\roman*)]
		\normalfont
		\item 
		\textit{
		$\Delta^0=0$
		}
		\normalfont
		\item \textbf{Locality and Translation covariance:}
		\textit{
		There exist linear maps $P_a^n : \mathscr{P}^{\otimes(n+1)}\rightarrow\mathscr{P}$ that are symmetric in the first $n$ entries and uniquely determined by 
		\begin{equation}
			\label{eq:locanom}
			\Delta^n\left(\otimes^n_{j=1} L_j(x_j) ;Q(y)
			\right)
			= \sum_{a\in(\mathbb{N}^d)^n} \partial^a 
			\delta(x_1-y,\dots,x_n-y)\,P_a^n(\otimes^n_{j=1} L_j;Q)(y)
		\end{equation}
		where the sum over $a$ is finite. 
		}
		\normalfont
		\item
		\textit{$\Delta^n\left(\otimes^n_{j=1} L_j(x_j) ;Q(y)
			\right) = \mathscr{O}(\hbar) \quad \forall n>0 \quad \text{if} \quad L_j \sim \hbar^0,\quad Q \sim \hbar^0$.
		}
		\normalfont
		\item \textbf{Field independence:}
		\textit{$\Delta^n$ depends on $\varphi$ only through its arguments.
		}
	\end{enumerate}
	\textit{If the $T$-product satisfies further renormalization conditions, these translate in the following way into properties of the maps $\Delta^n$}: 
	\begin{itemize}
		\item \normalfont{\textbf{Scaling degree:}} \textit{On the r.h.s of} (\ref{eq:locanom}) \textit{the sum over a is restricted by}
		\begin{equation}
			\label{eq:a_restr}
			|a| + \text{dim} \left(P_a^n(\otimes^n_{j=1} L_j;Q) \right)
			\leq \sum_{j=1}^{n}\text{dim}(L_j) + \text{dim}(Q)
			+ \frac{d+2}{2} - dn \; .
		\end{equation}
		\item \normalfont{\textbf{Lorentz covariance:}} 
		\begin{equation}
			\label{eq:anom_lorz}
			\beta_L \Delta(e^S_{\otimes};hQ) = 
			\Delta(e^{\beta_L S}_{\otimes};h\beta_L Q) \qquad \forall L \in \mathscr{L}_+^{\uparrow}
		\end{equation}
		\item \normalfont{\textbf{*-structure:}}
		\begin{equation}
			\label{eq:anomstar}
			\Delta(e^S_{\otimes};hQ)^* = 
			\Delta(e^{S^*}_{\otimes};\overline{h}Q^*)
		\end{equation}
		\item \normalfont{\textbf{off-shell field equation:}}
		\begin{equation}
			\label{eq:anomoffs}
			\Delta(e^S_{\otimes};h1) = 0
		\end{equation}
	\end{itemize}
\end{theorem}

This theorem states that the anomalous term is a local interacting field, that is $\Delta(e^S_{\otimes};hQ)$ is a local field and $R(e^{S/\hbar}_{\otimes}, \Delta(e^S_{\otimes};hQ))$ is the corresponding interacting field with interaction $S$ (see equation (\ref{eq:ret_MWI}), the additional $\hbar$ is present in the retarded quantum products). If this term cannot be removed completely from the anomalous MWI in (\ref{eq:anoMWI}) by finite renormalizations, it leads to additional interacting fields -- called quantum anomalies -- that were not present in the classical theory.\footnote{The most prominent experimentally measurable example of such an anomaly occurs in axial QED during the decay of the neutral $\pi^0$ meson \cite[chap. 5.3]{scharf95}.}

	\section{The case of a complex scalar field}
	\label{sec:complex}
	In this chapter we turn to the case of a complex scalar field, described by the two basic fields $\phi$ and $\phi^*$. Most definitions and all main results carry over in the expected ways. We will stick with the same symbols for the main building blocks of the theory, from now on they denote the complex counterparts to the real scalar theory.

\subsection{Basic definitions}

\begin{definition}
	The \textit{configuration space} is $\mathcal{E} := \mathcal{C}^{\infty} (\mathbb{M},\mathbb{C})$, the two basic fields are 
\begin{equation*}
	\phi(x),\phi^*(x): 
	\begin{cases}
	\mathcal{E} \rightarrow \mathbb{C}\\
	\phi(x)(h):=h(x), \quad \phi^*(x)(h):=\overline{h(x)} \; .
	\end{cases}
	\end{equation*}
	Let $\mathscr{F}'_{n,l}({\mathbb{M}^n})$ be defined analogously to $\mathscr{F}'({\mathbb{M}^n})$ but with symmetry required only among the first $l$ arguments and the following $n-l$ arguments separately. The \textit{field space} is defined as the set of all functionals $F\equiv F(\phi,\phi^*): \mathcal{E} \rightarrow \mathbb{C}$ of the form 
	\begin{equation*}
		F(\phi,\phi^*) = f_{0,0} + \sum_{n=1}^{N}\sum_{l=0}^{n} \int d x_1 \cdots d x_n \,
		\phi^*(x_1) \cdots \phi^*(x_l)\phi(x_{l+1})\cdots\phi(x_n) \,
		f_{n,l}(x_1,\dots,x_n)
	\end{equation*}
	with $f_{0,0} \in \mathbb{C}$ and $f_{n,l} \in\mathscr{F}'_{n,l}({\mathbb{M}^n})$ for $n \geq 1$. The $*$\textit{-conjugate field} $F^*$ of $F$ is defined as 
	\begin{equation*}
		F^*(\phi,\phi^*) = \overline{f_{0,0}} + \sum_{n=1}^{N}\sum_{l=0}^{n} \int d x_1 \cdots d x_n \,
		\phi(x_1) \cdots \phi(x_l)\phi^*(x_{l+1})\cdots\phi^*(x_n) \,
		\overline{f_{n,l}(x_1,\dots,x_n)} \; .
	\end{equation*}
	The \textit{functional derivative} is 
	\begin{align*}
		&\frac{\delta^{k+j} F}{\delta \phi^*(y_1) \cdots \delta \phi^*(y_k)\delta\phi(z_1) \cdots \delta \phi(z_j)} := 
		\sum_{n=k+j}^{N}\sum_{l=k}^{n-j} \frac{l!}{(l-k)!} \frac{(n-l)!}{(n-l-j)!} 
		\int d x_1 \cdots d x_{n-k-j} \, \notag\\
		&\qquad\qquad\qquad
		\cdot\phi^*(x_1) \cdots \phi^*(x_{l-k})
		\phi(x_{l-k+1}) \cdots \phi(x_{n-k-j}) \notag\\
		&\qquad\qquad\qquad
		\cdot f_{n,l}(y_1,\dots,y_k,x_1,\dots,x_{l-k},
		z_1,\dots,z_j,x_{l-k+1},\dots,x_{n-k-j})\; .
	\end{align*}
\end{definition}

So for derivatives w.r.t. $\phi$, the fields $\phi^*$ are treated as constants and vice versa. This is what we mean by saying that $\phi$ and $\phi^*$ are independent fields. 
\begin{definition}
	The \textit{star product} is
	\begin{align}
		\label{eq:complstar}
		F \star G &= 
		\sum_{n,k=0}^{\infty}\frac{\hbar^{n+k}}{n!k!}
		\int dx_1 ... dx_{n+k} \, dy_1 ... dy_{n+k} \notag\\
		&\cdot\frac{\delta^{n+k} F}{\delta \phi(x_1)...\delta\phi(x_n)\delta\phi^*(x_{n+1})...\delta\phi^*(x_k)}
		\;\prod_{l=1}^{n+k} \Delta_m^+(x_l - y_l)\notag\\
		&\cdot\frac{\delta^{n+k} G}{\delta \phi^*(y_1)...\delta\phi^*(y_n)\delta\phi(y_{n+1})...\delta\phi(y_k)} \; . 
	\end{align}
\end{definition}

This says that to compute star products among complex fields we need to take all contractions between pairs of $\phi$ and $\phi^*$. The product yields the basic commutators 
\begin{equation}
	\label{eq:comms}
	[\phi(x),\phi^*(y)] = [\phi(x)^*,\phi(y)] = i\hbar\Delta(x-y) \;, \qquad [\phi(x),\phi(y)] = [\phi^*(x),\phi^*(y)] = 0\; ,
\end{equation}
as expected for complex scalar theory. 

\begin{definition}
	The \textit{free action} is (again formally) given by 
	\begin{equation*}
		S_0 := \int dx \, L_0(x) =
		 \int dx \, \left(\partial_{\mu} \phi^*(x)\, \partial^{\mu} \phi(x) - m^2 \phi^* (x)\, \phi(x) \right) \; .
	\end{equation*}
	The corresponding free field equations are
	\begin{equation}
		\label{eq:feqcomp}
		\frac{\delta S_0}{\delta \phi(x)} := -(\dalemb +m^2) \phi^*(x)\;, \qquad 
		\frac{\delta S_0}{\delta \phi^*(x)} := -(\dalemb +m^2) \phi(x) \; .
	\end{equation}
	We consider a \textit{quartic interaction} of the form 
	\begin{equation*}
		S := -\kappa \int dx \,g(x)\, L_{\text{int}}(x) = -\kappa \int dx \, g(x) \left(\phi^*(x)\phi(x)\right)^2 \; .
	\end{equation*}
\end{definition}
We will frequently drop the subscript and write $L \equiv L_{\text{int}}$. The particular form of the interaction will be relevant only in chapter \ref{sec:proofMWI} to determine the arguments of the $T$-products for which the MWI will be shown to hold. In the following, all discussions treat the case of a free complex scalar field.

\subsection{Noether's Theorem for the complex scalar field}

In this section we apply Noether's theorem to the case of the complex scalar field, whose action is invariant under global phase transformations. 

\begin{definition}
	We define a \textit{global} $U(1)$\textit{-transformation} on the basic fields as 
	\begin{equation*}
		\phi(x) \mapsto \phi_{\alpha}(x) := e^{i\alpha}\phi(x)\; , 	\qquad 
		\phi^*(x) \mapsto \phi^*_{\alpha}(x) := e^{-i\alpha}\phi^*(x) \; , 
	\end{equation*}
	and on a general $F\in\mathscr{F}$ as
	\begin{equation*}
		F \mapsto F_{\alpha}(\phi,\phi^*) := F(\phi_{\alpha},\phi_{\alpha}^*) \;.
	\end{equation*}
\end{definition}

\begin{proposition}
	\label{prop:compneoth}
	Applying Noether's theorem to the free complex scalar field, that is to $S_{\text{\normalfont{tot}}}=S_0$ for this transformation as above we get
	\begin{align}
		s L_0 &= s L_{\text{\normalfont{int}}} = 0, \notag\\
		\qquad Q :&= Q_1 = s\phi= i\phi\;, \qquad Q_2 = s\phi^* = -i\phi^* = Q^*\;, \notag\\
		\label{eq:compcurrent}
		j^{\mu} &= i\,(\phi\,\partial^{\mu}\phi^* - \phi^*\, \partial^{\mu}\phi) \; .
	\end{align}
\end{proposition}

\begin{proof}
	Since the global phase transformation does not mix between the two basic fields, the generalization of Noether's theorem for the one field case is obviously done by just summing over both fields in equations (\ref{eq:noeth}) and (\ref{eq:compcurrent}). Furthermore, the transformation does not depend explicitly on $x$, so we do not need the test function in the interaction to be constant on any region of space time. The results follow by direct calculation. The total divergence $\Lambda^{\mu}=0$. 
\end{proof}

\subsection{Derivation of the classical MWI for a complex scalar field}

In this section we derive the classical MWI for a complex scalar field using Noether's theorem. We begin by defining a space of field polynomials which will be used frequently in the following. 

\begin{definition}
	\label{def:P1}
	Let $\mathscr{P}^{(1)} \subset \mathscr{P}$ be the space of field polynomials in $\phi,\phi^*,\partial^{\mu}\phi$ and $\partial^{\nu}\phi^*$ only. A basis element $P \in \mathscr{P}^{(1)}$ of this vector space may be written as 
	\begin{align}
		\label{eq:p0}
		P &= (\phi)^{\alpha_1}(\phi^*)^{\beta_1}
		(\partial^{\mu}\phi)^{\alpha_2}(\partial^{\nu}\phi^*)^{\beta_2}\;, \qquad \alpha_i,\beta_i \in \mathbb{N} \; , \notag\\
		a :&= \alpha_1 + \alpha_2\;, \qquad  b:=\beta_1 + \beta_2 \; .
	\end{align}
	So $a$ is the total number of $\phi$ and $\partial^{\mu}\phi$ while $b$ is the total number of $\phi^*$ and $\partial^{\mu}\phi^*$
\end{definition}

The following lemma gives a relation between the charge number operator and the derivation $\delta_{hQ}$ appearing in the MWI (\ref{eq:nthMWI}). 

\begin{lemma}
	\label{lemma:deltaQ}
	Let $P \in \mathscr{P}^{(1)}$. Define the corresponding charge number operators
	\begin{align}
		\label{eq:theta}
		\theta &:= 
		\phi\frac{\partial}{\partial \phi}
		+\partial^{\mu}\phi\frac{\partial}{\partial(\partial^{\mu}\phi)}
		- \phi^*\frac{\partial}{\partial \phi^*}
		-\partial^{\mu}\phi^*\frac{\partial}{\partial(\partial^{\mu}\phi^*)} \; , \notag\\
		\theta_{\mu} &:= 
		\phi\frac{\partial}{\partial(\partial^{\mu}\phi)}
		-\phi^*\frac{\partial}{\partial(\partial^{\mu}\phi^*)} \;.
	\end{align}
	Then 
	\begin{equation*}
		\delta_{Q(y)}P(x) =
		-i \Big(\delta(y-x)\left(\theta P\right)(x)
		-\partial_y^{\mu}\big(\delta(y-x)\left(\theta_{\mu}P\right)(x)\big)\Big)\;, 
	\end{equation*}
	where $	\delta_{Q(y)}$ is as in equation \normalfont{(\ref{eq:derQ})} \textit{with the test function $h$ omitted.}
\end{lemma}
\begin{proof}
	For the case of two fields, we have to sum over both of them in the definition of the derivation in equation (\ref{eq:derQ}), that is $\delta_{hQ}= \sum_i \delta_{h_i Q_i}$. Writing down only the integral kernel (by dropping $h$) and using the particular form of $Q_i$ from proposition \ref{prop:compneoth} we calculate
	\begin{align*}
		\delta_{Q(y)}P(x) &= \Big(Q(y) \frac{\delta}{\delta \phi(y)} + Q^*(y) \frac{\delta}{\delta \phi^*(y)}\Big) P(x) \notag\\
		&= -i\phi(y) \Big(\delta(y-x)\frac{\partial P}{\partial \phi}(x) - \partial_y^{\mu}\delta(y-x)\frac{\partial P}{\partial (\partial \phi)}(x)
		\Big) \notag\\
		&\qquad + i\phi^*(y) \Big(\delta(y-x)\frac{\partial P}{\partial \phi^*}(x) - \partial_y^{\mu}\delta(y-x)\frac{\partial P}{\partial (\partial \phi^*)}(x)
		\Big) \notag\\
		&= -i \delta(y-x) \Bigg(
		 \phi(y)\frac{\partial P}{\partial \phi}(x)
		+\partial^{\mu}_y \phi(y)\frac{\partial P}{\partial(\partial^{\mu} \phi)}(x) \notag\\
		& \qquad\qquad\qquad\qquad-\phi^*(y)\frac{\partial P}{\partial \phi^*}(x)
		- \partial^{\mu}_y \phi^*(y)\frac{\partial P}{\partial(\partial^{\mu} \phi^*)}(x)
		\Bigg) \notag\\
		&\qquad\qquad + i \partial^{\mu}_y \Bigg[ \delta(y-x) \Big(
		\phi(y)\frac{\partial P}{\partial(\partial^{\mu} \phi)}(x)
		- 	\phi^*(y)\frac{\partial P}{\partial(\partial^{\mu} \phi^*)}(x)
		\Big)\Bigg]\notag\\
		&= 	-i\delta(y-x)\left(\theta P\right)(x)
		+i\partial_y^{\mu}\big(\delta(y-x)\left(\theta_{\mu}P\right)(x)\big)\; ,
	\end{align*}
	where we have used the properties of the equation for the functional derivative (\ref{eq:funcder}) as well as the chain and product rule for the $\delta$-distribution. 
\end{proof}

\begin{remark}
	\label{rem:thetaleibn}
	For later purposes, we note that $\theta$ as introduced in lemma \ref{lemma:deltaQ} satisfies the Leibniz rule for the complex $\star$-product 
	\begin{equation*}
		\theta (F \star G) = (\theta F) \star G + F \star (\theta G) \qquad \forall F,G \in \mathscr{F} \; .
	\end{equation*}
	This may be verified by direct calculation, see \cite[Chap. 5.1.4]{duetsch19}
\end{remark}

The following proposition translates the general MWI in the form (\ref{eq:nthMWI}) into the case of the complex scalar field using Noether's theorem. The resulting MWI is the one to be shown to hold for the quantum case in chapter \ref{sec:proofMWI}. 

\begin{proposition}
	\label{prop:compMWI}
	For $P_1,\dots,P_n \in \mathscr{P}^{(1)}$ polynomials in the two basic fields and their first derivatives the MWI for a complex scalar field can be written as 
	\begin{align}
		\label{eq:compMWI}
		\partial_{\mu}^y \,T_{n+1}\big( P_1(x_1)&\otimes\cdots\otimes P_n(x_n)\otimes j^{\mu}(y) \big) 
		\notag\\
		&-\hbar\sum_{l=1}^{n}\delta(y-x_l)\, 
		T_{n}\big( P_1(x_1)\otimes\cdots\otimes(\theta P_l)(x_l)\otimes\dots\otimes P_n(x_n) \big)
		 \notag\\
		& +\hbar\,\partial_y^{\mu} \Big(
		\sum_{l=1}^{n}\delta(y-x_l)\,
		T_{n}\big( P_1(x_1)\otimes\cdots\otimes(\theta_{\mu} P_l)(x_l)\otimes\dots\otimes P_n(x_n) \big)\Big) 
		\notag\\
		= 
		i\, &T_{n+1}\big( P_1(x_1)\otimes\cdots\otimes P_n(x_n)\otimes \phi(y) \big) \cdot(\dalemb\,+\,m^2)\phi^*(y) 
		\notag\\
		-i\, &T_{n+1}\big( P_1(x_1)\otimes\cdots\otimes P_n(x_n)\otimes \phi^*(y) \big) \cdot(\dalemb\,+\,m^2)\phi(y) \; .
	\end{align}
\end{proposition}
\begin{proof}
	We start with equation (\ref{eq:nthMWI}), the MWI for the $T$-products to $n$-th order for off-diagonal entries. With the
	functional $A$ from (\ref{eq:funcA}) and the derivation $\delta_{hQ}$ from (\ref{eq:derQ}), both for the case of two basic fields with the sum over $i$ running through $\phi_i=\phi$, $\phi_2=\phi^*$ we get
	\begin{align*}
		&\int dy \, \sum_{i}h_i(y)\, T_{n+1} \left(F_1\otimes \cdots F_n \otimes Q_i(y)\,\frac{\delta S_0}{\delta \phi_i(y)}\right) \notag\\
		&\qquad+ \int dy \, \sum_{i}h_i(y)\cdot \frac{\hbar}{i} \sum_{l=1}^{n} 
		T_{n} \left(F_1\otimes \cdots \otimes \delta_{Q_i(y)} F_l \otimes \cdots \otimes F_n\right) \notag\\
		= &\int dy\, \sum_{i}h_i(y)\,
		T_{n+1} \left(F_1\otimes \cdots F_n \otimes Q_i(y)\right) 
		\cdot \frac{\delta S_0}{\delta \phi_i(y)} \; .
	\end{align*}
	On the r.h.s. we use the free field equations (\ref{eq:feqcomp}). On the l.h.s, in the first term we replace the divergence of the Noether current from equation (\ref{eq:noeth}) for the free action (since it is $S_0$ that appears in $A$). Now we choose $F_1,\dots,F_n$ to be local functionals of the form $F_i = \int dx\, g_i(x) P_i(x)$ for $P_i\in \mathscr{P}^{(1)}$. In the second term on the l.h.s we can then put in the form of $\delta_{Q(y)}$ from lemma \ref{lemma:deltaQ}. Finally by using the Action Ward Identity (\ref{eq:AWI}) to pull the derivative out of the first $T$-product on the l.h.s and omitting all test functions $h_i$ and $g_i$ we arrive at equation (\ref{eq:compMWI}). 
\end{proof}	
	
	\chapter{Proving the MWI for the complex scalar field}
	\label{chap:proof}
	
	\section{Preliminaries to the proof}
	\label{sec:prelims}
	This section presents four statements that will be needed for the proof of the MWI in the subsequent section. They are all given in \cite{duetsch19} either for the scalar or the QED case, but for two of them the proof has to be modified to fit the complex scalar field. 

\subsection{Anomaly with basic fields as arguments}

The following proposition will be used to reduce the number of arguments of the anomaly map $\Delta^k$ that have to be discussed in section \ref{sec:6case}. It states that in the inductive procedure of the proof performed in chapter \ref{sec:proofMWI}, anomaly maps with arguments consisting of only one basic field vanish. 

\begin{proposition}
	\label{prop:argums}
	Let $Q\in\mathscr{P}$. Assume that $T_1,\dots,T_n$ are constructed and that \normalfont
	\begin{equation*}
		\Delta^k \left( \otimes_{j=1}^{n-1} F_j;Q(y)\right) = 0
		\quad \forall F_j \in \mathscr{F}_{\text{loc}}, \quad k<n\; .
	\end{equation*}
	\textit{Then }
	\begin{equation*}
		\Delta^k \Big(\partial^a\phi_i(x)\otimes\bigotimes_{j=1}^{n-1} F_j;Q(y)\Big) = 0
		\quad \forall F_j \in \mathscr{F}_{\text{loc}} \; , \quad \forall a\in\mathbb{N}^d \; .
	\end{equation*}\textit{
	for each of the basic fields $\phi_1=\phi,\phi_2=\phi^*$. }
\end{proposition}

\begin{proof}
	This follows by using the off-shell field equation, see \cite[Exc. 4.3.3]{duetsch19} for the real scalar case $\varphi$. Repeating exactly the same calculation for $\phi^*$ and for fields with derivatives $\partial^a$ yields this statement for the complex scalar field.
\end{proof}

\subsection{The charge number operator}

We now show that the action of the charge number operator $\theta$ can be expressed as taking a commutator involving the zeroth component of the current $j^{\mu}$. The proof will use a Lorentz invariant version of the Gauss integral theorem taken from \cite[chap. A.1]{duetsch19}, that is stated in the following lemma. 

\begin{lemma}
	Let $v\in\mathcal{C}^1(\mathbb{M},\mathbb{R}^4)$ be a vector field and $G\subset\mathbb{M}$ a compact set with a sufficiently smooth boundary $\partial G$. Then 
	\begin{equation*}
		\int_G d^4x\, \partial_{\mu} v^{\mu}(x) = \oint_{\partial G} d\sigma_{\mu}(x)\,v^{\mu}(x)
	\end{equation*} 
	for some measure $d\sigma$. For the special case that in some region the boundary $\partial G$ is of the form $x^0 = \text{\normalfont{const}}$, in this region we have 
	\begin{equation*}
		d\sigma = (\pm1,0,0,0)\,d\vec{x}
	\end{equation*}
	with the sign chosen in a way that $d\sigma$ is pointing outwards of $G$.
\end{lemma}

\begin{proposition}
	\label{prop:Q_0}
	We formally write 
	\begin{equation}
		\label{eq:Q_0}
		Q_0 := \int d\vec{x}\, j^0(c,\vec{x})_0 \qquad \text{\normalfont{with}} \quad c\in\mathbb{R} \; .
	\end{equation}
	Then 
	\begin{enumerate}[label=(\roman*)]
	\normalfont
	\item 
	\textit{For $F\in\mathscr{F}$, the commutator $[Q_0,F_0]$ exists and does not depend on $c$. 
	}
	\normalfont
	\item 
	\textit{For $P\in\mathscr{P}^{(1)}$ we have 
	\begin{equation}
		\label{eq:chargecomm}
		[Q_0,P(x)_0]= \hbar\,(\theta P)(x)_0 \; .
	\end{equation}
	}
	\end{enumerate}
\end{proposition}

\begin{proof}
	(i) It is a priori not clear whether the expression in equation (\ref{eq:Q_0}) exists, but we will only consider its commutator with $F_0$. Due to spacelike commutativity (proposition \ref{prop:spacecomm}) the commutator of two fields $F,G\in\mathscr{F}$ vanishes, if the supports of $F$ and $G$ are spacelike separated. So 
	\begin{equation*}
		[j^0(c,\vec{x})_0,F_0]=0 \qquad \text{if}\;
		(c,\vec{x}) \notin \Big( \text{supp}\,F + (\overline{V}_+ \cup \overline{V}_-)\Big)
	\end{equation*}
	and -- since $\text{supp}\,F$ is bounded -- the region of integration in (\ref{eq:chargecomm}) is bounded, so the integral and thus the commutator exist. \\
	
	\begin{figure}[h!]
		\centering
		\includegraphics[width=0.8\textwidth]{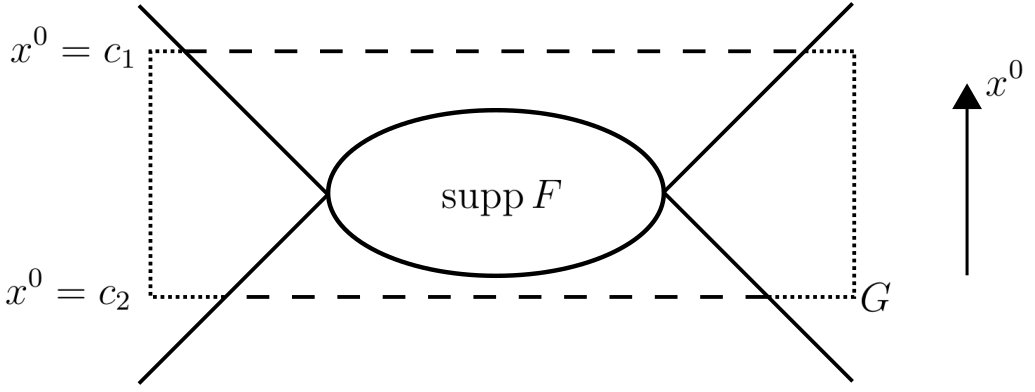}
		\caption[Testbild]{The region $G$ of integration used to show independence of $c\in\mathbb{R}$ of the commutator $[Q_0,F_0]$ in proposition \ref{prop:Q_0}. \label{fig:gauss}}
	\end{figure}
	To show independence of $c$ we choose a region $G$ to integrate over the commutator as indicated in figure \ref{fig:gauss}. Due to equation (\ref{eq:noeth}) the divergence of the on-shell current vanishes. Using Gauss Theorem and spacelike commutativity we calculate 
	\begin{align*}
		0 
		&= \int_G d^4x\, \partial_{\mu}^x \,[j^{\mu}(x)_0,F_0] 
		= \oint_{\partial G} d\sigma_{\mu}\, [j^{\mu}(x)_0,F_0] \notag\\
		&= \int_{x^0=c_1} d\vec{x}\,[j^{0}(x)_0,F_0] -  
		\int_{x^0=c_2} d\vec{x}\,[j^{0}(x)_0,F_0] \notag\\
		&= [Q_0,F_0]_{c_1} - [Q_0,F_0]_{c_2} \; .
	\end{align*}
	Hence the commutator is independent of the choice of $c$. \\

	(ii) Using the particular form of $j^{\mu}$ in (\ref{eq:compcurrent}), the basic commutators (\ref{eq:comms}) and proposition \ref{prop:retprop} we calculate for an arbitrary $c$: 
	\begin{align*}
		[Q_0,\phi(y)_0] &= i \int_{x^0=c}d\vec{x}\, \Big[\phi(x)_0\,\partial^0 \phi^*(x)_0 - \phi^*(x)_0\,\partial^0 \phi(x)_0,\phi(y)_0\big] \notag\\
		&= i \int_{x^0=c}d\vec{x}\,\Big(\phi(x)_0\, i \hbar\, \partial_x^0 \Delta(x-y) - (\partial_x^0 \phi(x)_0)\, i\hbar\, \Delta(x-y) 
		\big) \;, \notag\\
		&= \hbar \int_{x^0=y^0} \phi(y^0,\vec{x})\,\delta(\vec{x}-\vec{y}) \notag\\
		&= \hbar\, \phi(y)_0 \; , \notag\\ \\
		[Q_0,\phi(y)_0] &= -\hbar\, \phi^*(y)_0 \;.
	\end{align*}
	In the third line we have chosen $c=y^0$, which is possible since the result is independent of $c$. The calculation for $\phi^*$ is performed following precisely the same steps. Using linearity of the derivative and the commutator we also get 
	\begin{equation*}
		[Q_0,\partial^{\mu}\phi(y)_0] = \hbar\, \partial^{\mu}\phi(y)_0 \; , \qquad
		[Q_0,\partial^{\mu}\phi^*(y)_0] = - \hbar\, \partial^{\mu}\phi^*(y)_0 \; .
	\end{equation*}
	Now let $P\in\mathscr{P}^{(1)}$ be a basis element. Then due to the derivation property of the commutator (theorem \ref{thm:starprod}) we get 
	\begin{align*}
		[Q_0,P(x)_0] &= \hbar \,(a-b) P(x)_0\notag\\
					 &= \hbar \, (\theta P)(x)_0\; ,
	\end{align*}
	where we use the notation introduced in definition \ref{def:P1}. By linearity of the commutator, the statement follows for any $P\in\mathscr{P}^{(1)}$. 
\end{proof}

\subsection{Furry's Theorem}

Furry's theorem will be used to conclude that the VEVs of certain $T$-products have to vanish. It makes use of the notion of charge conjugation, which is essentially the operation of exchanging all $\phi$ and $\phi^*$ by each other. 

\begin{definition}
	The \textit{charge conjugation operator} $\beta_C$ on the basic fields is defined as 
	\begin{equation*}
	\beta_C \, \phi(x) := \eta_C \, \phi^* (x) \; ,\qquad 
	\beta_C \, \phi^* (x) := \eta_C^* \, \phi (x) \; , \qquad \eta_C \in \mathbb{C}\; \; \text{with} \; \; 
	|\eta_C| = 1 \; ,
	\end{equation*}
	and as an operator on the space of complex scalar fields $\beta_C: \mathscr{F} \rightarrow \mathscr{F}$ by
	\begin{equation*}
	\beta_C \, F (\phi,\phi^*) :=  F (\beta_C \,\phi,\beta_C \,\phi^*).
	\end{equation*}
\end{definition}

\begin{proposition}
	Charge conjugation $\beta_C$ is a linear operator satisfying $\beta_C^2 = 1$ and $\beta_C(F\cdot G) = (\beta_C F)\cdot(\beta_C G)$ for $F,G \in \mathscr {F}$. Furthermore it holds that
	\begin{equation}
		\label{eq:betastar}
		\beta_C(F\star G) = (\beta_C F) \star (\beta_C G) \; . 
	\end{equation}
\end{proposition}

\begin{proof}
	Linearity and the first two properties are obvious from the definition of $\beta_C$. To prove the relation with the $\star$-product we first calculate 
	\begin{equation*}
	\frac{\delta(\beta_C F)}{\delta (\beta_C \phi(x))}
	= \beta_C\left( \frac{\delta F}{\delta \phi(x)} \right)\;, \qquad 
	\frac{\delta(\beta_C F)}{\delta (\beta_C \phi^*(x))}
	= \beta_C\left( \frac{\delta F}{\delta \phi^*(x)} \right)\;.
	\end{equation*}
	Considering one contribution to equation (\ref{eq:betastar}) at order $\hbar$ (corresponding to $n=1,k=0$ in equation (\ref{eq:complstar})) we find 
	\begin{equation*}
		\beta_C \left(\frac{\delta F}{\delta \phi(x)}\right)
		\Delta^+(x-y)
		\beta_C \left(\frac{\delta G}{\delta \phi(x)}\right)
		= \frac{\delta (\beta_C F)}{\delta \phi^*(x)}
		\Delta^+(x-y)
		\frac{\delta (\beta_C G)}{\delta \phi(x)^*} \; ,
	\end{equation*}
	which is precisely the second contribution to order $\hbar$ ($n=0,k=1$ in equation (\ref{eq:complstar})). Hence at order $\hbar$, $\beta_C$ exchanges contributions to the star product with their complex conjugates, leaving the overall sum unchanged. The result now follows by induction on the order in $\hbar$. 
\end{proof}

As a further axiom for the $T$-product we require
\begin{enumerate}[label=(\roman*)]
	\setcounter{enumi}{11}
	\item \textbf{charge conjugation invariance:}
	\begin{equation*}
		\beta_C \circ T_n = T_n \circ \beta_C^{\otimes n} \; .
	\end{equation*}
\end{enumerate}

\begin{proposition}
	Charge conjugation invariance is a renormalization condition, which may be satisfied while preserving all other axioms and renormalization conditions. 
\end{proposition}

\begin{proof}
	To show that the condition holds true for unrenormalized $T$-products one may use property (\ref{eq:betastar}) of the charge conjugation operator.  For the complete proof see \cite[chap. 5.1.5]{duetsch19}. 
\end{proof}

Now we give our version of Furry's theorem. 

\begin{theorem}
	\label{thm:gfurry}
	Let $A_i$, $B_j, i = 1,\dots,r, j = 1,\dots,s$ be fields in $\mathscr{F}$ with $\beta_C A_i = A_i$ and $\beta_C B_j = -B_j$ for all $i,j$. Then for $r\in\mathbb{N}$
	\begin{equation}
		\label{eq:gfurry}
		t(A_1,\dots,A_r,B_1,\dots,B_s) = 0 \qquad 
		\text{\normalfont{if $s$ is odd.}}
	\end{equation}
\end{theorem}

\begin{proof}
Since the VEV of a $T$-product picks out the term $f_{0,0}$ with no field operators, we have $\omega_0 \circ \beta_C = \omega_0$. Using the charge conjugation invariance axiom and linearity of the $T$-product yields 

\begin{align*}
	&t(A_1,\dots,A_r,B_1,\dots,B_s) \notag\\
	&= \omega_0 \Big( T_n\big(A_1,\dots,A_r,B_1,\dots,B_s\big)
	\Big) \notag\\
	&= \omega_0 \Big(\beta_C T_n\big(A_1,\dots,A_r,B_1,\dots,B_s\big)
	\Big) \notag\\
	&= \omega_0 \Big(T_n\big(\beta_C A_1,\dots, \beta_C A_r,\beta_C B_1,\dots,\beta_C B_s\big)
	\Big) \notag\\
	& = (-1)^s \;t(A_1,\dots,A_r,B_1,\dots,B_s) \; .
\end{align*}
Hence if $s$ is odd, the VEV has to vanish for arbitrary $r$.
\end{proof}

\begin{corollary}
	\label{corr:furry}
	Choosing $B_i = j^{\mu} = i(\phi\,\partial^{\mu}\phi^* - \phi^*\, \partial^{\mu}\phi)$ and $A_j = L = (\phi^* \phi)^2$ for all $i$ and $j$, which both satisfy the charge conjugation conditions, we get Furry's theorem for complex $\phi^4$ theory: 
	\begin{equation*}
	t(j^{\mu_1},...,j^{\mu_s},L,...,L) = 0 \qquad 
	\text{\normalfont{if $s$ is odd.}}
	\end{equation*}
\end{corollary}

\subsection{A version of the Poincaré lemma}

This proposition states a version of the Poincaré lemma for local field polynomials. It will be used to restrict the possible terms contributing to the anomaly in the next chapter. 

\begin{proposition}
	Let 
	\begin{equation*}
		f(y,x_1,\dots,x_n) = 
		\sum_{a\in (\mathbb{N}^d)^n} \partial^a 
		\delta(x_1-y,\dots,x_n-y)\, P_a(y) \qquad \text{\normalfont{with}}\quad P_a \in \mathscr{P}
	\end{equation*}
	and 
	\begin{equation*}
		\int dy \, f(y,x_1,\dots,x_n)= 0\; .
	\end{equation*}
	Then there exist polynomials $U_a^{\mu} \in \mathscr{P}$ such that 
	\begin{equation}
		\label{eq:poinc}
		f(y,x_1,\dots,x_n) = \partial_{\mu}^y \Bigg(
		\sum_{a\in (\mathbb{N}^d)^n} \partial^a 
		\delta(x_1-y,\dots,x_n-y)\, U_a^{\mu}(y) \Bigg)\; .
	\end{equation}
\end{proposition}

\begin{proof}
	See \cite[Lemma 4.5.1]{duetsch19}. 
\end{proof}

	\section{The proof of the MWI}
	\label{sec:proofMWI}
	
In this chapter we will prove the MWI for the complex scalar field for a certain class of arguments $P_1,...,P_n$, namely the elements of the following space:  
\begin{definition}
	Let $\mathscr{P}_{(\phi^*\phi)^2}\subset \mathscr{P}^{(1)} \subset \mathscr{P}$ be the vector space of polynomials spanned by the basis elements
	\begin{equation*}
		\mathscr{B}_{(\phi^*\phi)^2} := \big\{(\phi)^{m}(\phi^*)^{n},
		\partial^{\mu}\phi, \partial^{\nu}\phi^*, j^{\eta}
		\big\} \; , \qquad 0 \leq m, n \leq 2 \; ,
	\end{equation*} 
	that is the current, the quartic complex interaction $L$ and all their submonomials.  
\end{definition}

The purpose of this choice is that by showing the MWI for these monomials, we can express all the $T$-products of the form $T(L,\dots,L,j,\dots,j)$ via their causal Wick expansion, in which $t$-products with submonomials of $L$ and $j$ as arguments will appear. We will adapt the proof given by Dütsch and Fredenhagen in \cite{duetschfred99} for the MWI of QED. The basic idea of the proof and the essential steps carry over to the scalar case. Major modifications are mainly due to the fact that the scalar current in (\ref{eq:compcurrent}) contains derivatives of the basic fields, which the QED current $j^{\mu} = \overline{\psi} \wedge \gamma^{\mu} \psi$ doesn't. This leads to an additional term in the complex scalar MWI -- the total divergence in the third line of equation (\ref{eq:compMWI}) -- which has to be taken into account throughout the proof. The main theorem of this thesis thus is the following.

\begin{theorem}
	\label{thm:bachelor}
	For all $n\in\mathbb{N}$, the $T_{n}$ can be renormalized in a way that the complex scalar MWI in equation \normalfont{(\ref{eq:compMWI})}\textit{ holds true for all $P_1,\dots,P_n \in \mathscr{P}_{(\phi^*\phi)^2}$.}
\end{theorem}

To prove this theorem we need to show that the anomaly map $\Delta$ vanishes to all orders for arbitrary arguments in $\mathscr{P}_{(\phi^*\phi)^2}$. Since the $T$-products and $\Delta$ are both linear, it is sufficient to show this for all basis elements in $\mathscr{B}_{(\phi^*\phi)^2}$. Our proof will proceed by induction on $n$ as follows. In section \ref{sec:1base} we provide the basis of the induction by showing that the anomalous term to $0$-th order is $\Delta^0=0$, and we give an expression for the $n$-th order term. Section \ref{sec:2charge} shows that the $T_n$ can be renormalized such that they satisfy charge number conservation. In section \ref{sec:3integr} it is shown that the integral w.r.t. the $y$ entry of the anomalous term vanishes, from what certain properties of this term can be deduced. Section \ref{sec:4struc} translates the results of the anomalous MWI (theorem \ref{thm:aMWI}) into properties of the VEVs of the anomalous terms. In the next section \ref{sec:5renom} it is shown that almost all anomalies can be removed by finite renormalizations that are compatible with all axioms and renormalization conditions except for a few that exhibit certain symmetries. The last section \ref{sec:6case} works out the specific anomalies for these cases and shows that admissible renormalizations removing the anomalies exist.

\subsection{Base case of induction}
\label{sec:1base}

This section proves the following proposition, which provides the base case of the induction and states what is to be shown in the inductive step. 

\begin{proposition}
	We have $\Delta^0=0$. Furthermore, assume that the quantum-MWI is satisfied to orders $k<n$ in the fields $P_j$. Then the anomalous term to order $n$ is
	\begin{align}
		\label{eq:start}
		\frac{\hbar^n}{i^n}
		\Delta^n&\left(P_1(x_1) \otimes \dots \otimes  P_n(x_n);Q(y)\right) 
		= 
		- \partial_{\mu}^y \,T_{n+1}\big( P_1(x_1)\otimes\cdots\otimes P_n(x_n)\otimes j^{\mu}(y) \big) \notag\\
		&+\hbar\sum_{l=1}^{n}\delta(y-x_l)\, 
		T_{n}\big( P_1(x_1)\otimes\cdots\otimes(\theta P_l)(x_l)\otimes\dots\otimes P_n(x_n) \big)
		\notag\\
		& -\hbar\,\partial_y^{\mu} \Big(
		\sum_{l=1}^{n}\delta(y-x_l)\,
		T_{n}\big( P_1(x_1)\otimes\cdots\otimes(\theta_{\mu} P_l)(x_l)\otimes\dots\otimes P_n(x_n) \big)\Big) 
		\notag\\
		&+i\, T_{n+1}\big( P_1(x_1)\otimes\cdots\otimes P_n(x_n)\otimes \phi(y) \big) \cdot(\dalemb\,+\,m^2)\phi^*(y) 
		\notag\\
		&-i\, T_{n+1}\big( P_1(x_1)\otimes\cdots\otimes P_n(x_n)\otimes \phi^*(y) \big) \cdot(\dalemb\,+\,m^2)\phi(y) \; .
	\end{align}
\end{proposition}

\begin{proof}
	We start with equation (\ref{eq:anoMWI}), the anomalous MWI for the $T$-products, but for the case of two fields $\phi$ and $\phi^*$. This implies modifying $A$ and $\delta_{hQ}$ as in the proof of proposition \ref{prop:compMWI} by summing over both fields. Since $Q_1 = Q_2^*=Q$, we use only $Q$ as argument of the anomaly map $\Delta(\dots;Q)$. Now we write out (\ref{eq:anoMWI}) to a fixed order $n$ in the coupling constant $\kappa$ for non-diagonal entries $S_i=\kappa L_i(g_i)$ omitting all test functions
	\begin{align*}
		\Delta^n&\left(\otimes_{j=1}^n L_j(x_j);Q(y)\right)
		= \left(\frac{i}{\hbar}\right)^n T \Big(\bigotimes_{j=1}^n L_j(x_j)\otimes Q_i(y)\Big)
		\cdot \frac{\delta S_0}{\delta \phi_i(y)} \notag\\ 
		&- \left(\frac{i}{\hbar}\right)^n 
		T \Big(\bigotimes_{j=1}^n L_j(x_j)\otimes Q_i(y)\cdot
		\frac{\delta S_0}{\delta \phi_i(y)}\Big) \notag\\
		&- \sum_{l=1}^{n}
		\left(\frac{i}{\hbar}\right)^{n-1} 
		T \Big(\bigotimes_{j=1, j\neq l}^n L_j(x_j)\otimes \delta_{Q_i(y)}L_l(x_l)\Big) \notag\\
		&- \sum_{I \subset \{1,\dots,n\}, I^c \neq \emptyset}
		\left(\frac{i}{\hbar}\right)^{|I^c|} 
		T\Big(\bigotimes_{k\in I^c} L_k(x_k)\otimes
		\Delta^{|I|}\big(\otimes_{j\in I} L_j(x_j);Q(y)\big)\Big) \; ,
	\end{align*}
	where the sum over $Q_i$ and $\phi_i$ with $i \in \{1,2\}$ is implied. If $n=0$, only the first two terms on the r.h.s contribute and they cancel each other, hence $\Delta^0 = 0$. Due to the inductive assumption, we have $\Delta^k =0$ for $k<n$, so the term in the last line vanishes. Now we assume the $L_j$ to be polynomials $P_j \in\mathscr{P}_{(\phi^*\phi)^2}$ and proceed exactly as in the proof of proposition \ref{prop:compMWI}, by putting in the field equations, the expression for $\delta_Q$, the current $j^{\mu}$ and using the AWI. After finally multiplying with $\hbar^n / i^n$ we arrive at (\ref{eq:start}). 
\end{proof}

\subsection{Charge number conservation}
\label{sec:2charge}

In this section we show charge number conservation for the $T$-products.
\begin{proposition}
	\label{prop:Tchc}
	In the inductive construction, the $n$-th order $T$-product with arguments in $\mathscr{P}^{(1)}$ that are eigenvectors of the charge number operator $\theta$ may be renormalized such that it satisfies charge number conservation 
	\begin{equation}
	\label{eq:chconsv}
	\theta\, T\big(P_1(x_1) \otimes ... \otimes  P_n(x_n)\big) 
	= T\big(P_1(x_1) \otimes ... \otimes  P_n(x_n)\big) \cdot
	\sum_{j=1}^{n} (a_j-b_j) \; .
	\end{equation}
\end{proposition}

\begin{remark}
	The restriction to $\mathscr{P}^{(1)}$ is not necessary to satisfy charge number conservation, but due to our definition of the charge number operator in (\ref{eq:theta}) we show this particular case. The proof for arbitrary $P_i$ follows exactly the same path. 
\end{remark} 

\begin{proof}
	We start by proving the following statement: 
	\begin{lemma}
		\label{lemm:pups}
		Let $P_1,\dots,P_n \in \mathscr{P}^{(1)}$ \normalfont{(}\textit{see definition} \normalfont{\ref{def:P1}}) \textit{be eigenvectors of $\theta$ and assume that the $T$ products to order $n-1$ satisfy the property} \normalfont{(\ref{eq:chconsv}).}\textit{ We then may renormalize the $n$-th order $t$-product in a way that the following implication holds true: If $t\big(P_1,\dots,P_n\big) \neq 0$, then}
		\begin{equation*}
			\sum_{j=1}^{n} (a_j-b_j) = 0 \; .
		\end{equation*}
	\end{lemma}
	
	\begin{proof}
		We first show that the statement holds true for unrenormalized $t$-products $t_{\text{unr}}$. In the inductive construction, the expression for $T(x_1,\dots,x_n)$ factorizes causally on points laying outside the thin diagonal. As noted in remark \ref{rem:thetaleibn}, $\theta$ satisfies the Leibniz rule for the $\star$-product and hence also for the Feynman star product $\star_F$. By using this and the inductive assumption on charge number conservation, we see that (\ref{eq:chconsv}) holds true for unrenormalized $T$-products. Now since the charge number operator $\theta$ acts as a derivative operator, we have $\omega_0 \circ \theta =0$. Using this in (\ref{eq:chconsv}) we calculate 
		\begin{align*}
			0 = \omega_0 \circ \theta \,\big( T_{\text{unr}}(P_1,\dots,P_n) \big) 
			= t_{\text{unr}}\big( P_1,\dots,P_n \big) \cdot \sum_{j=1}^{n} (a_j-b_j)  \; . 
		\end{align*}
		This shows the lemma for all unrenormalized $t$-products. Now when renormalizing these expressions, the implication can only get lost if for some unrenormalized $\hat{t}_{\text{unr}} = 0$ with $\sum_{j=1}^{n} (a_j-b_j) \neq 0$, the corresponding renormalized $\hat{t}\neq 0$. So we just extend all vanishing unrenormalized $t$-products for which $\sum_{j=1}^{n} (a_j-b_j) \neq 0$ by zero. This is compatible with all other normalization conditions and completes the proof. 
	\end{proof}
	To complete the proof of proposition \ref{prop:Tchc} we use the causal Wick expansion (\ref{eq:cwick})
	\begin{equation}
	\label{eq:abc}
	T\left( P_1(x_1)\otimes\cdots\otimes P_n(x_n) \right)
	= \sum_{\underline{P}_l\subseteq P_l}t\left( \underline{P}_1(x_1),\dots,\underline{P}_n(x_n)
	\right)
	\overline{P}_1(x_1)\cdots\overline{P}_n(x_n) \; . 
	\end{equation}
	Let $\underline{a}_j$ be the total number of factors $\phi$ and $\partial_{\mu}\phi$ contained in $\underline{P}_j$. Define $\overline{a}_j$ analogously for $\overline{P}_j$ as well as $\underline{b}_j$ and $\overline{b}_j$ for the number of factors of $\phi^*$ and $\partial_{\mu}\phi^*$. It then holds that 
	\begin{equation*}
		\underline{a}_j + \overline{a}_j = a_j \;, \qquad 
		\underline{b}_j + \overline{b}_j = b_j \; .
	\end{equation*}
	From lemma \ref{lemm:pups} we know that all the non vanishing terms on the r.h.s of (\ref{eq:abc}) satisfy 	$\sum_{j=1}^{n} (\underline{a}_j-\underline{b}_j) = 0$, so we may just add them in the following calculation 
	\begin{align*}
		\theta \Big( t\left( \underline{P}_1,\dots,\underline{P}_n
		\right)
		\overline{P}_1\cdots\overline{P}_n
		\Big)
		&= \Big( t\left( \underline{P}_1,\dots,\underline{P}_n
		\right)
		\overline{P}_1\cdots\overline{P}_n
		\Big)\cdot\sum_{j=1}^{n} (\overline{a}_j-\overline{b}_j) \notag \\
		&=\Big( t\left( \underline{P}_1,\dots,\underline{P}_n
		\right)
		\overline{P}_1\cdots\overline{P}_n
		\Big)\cdot\sum_{j=1}^{n} ({a}_j-{b}_j) \; .
	\end{align*}
	Hence we have charge number conservation for all individual terms on the r.h.s. of (\ref{eq:abc}) satisfied, which shows overall charge number conservation for the $T$-product. 
\end{proof}

\subsection{The integrated anomalous term vanishes}
\label{sec:3integr}

We show the following proposition, which will allow us to make further statements about the structure of the anomalous term. 

\begin{proposition}
	The integral over the last argument of the on-shell anomalous term vanishes, that is 
	\begin{equation}
		\label{eq:int0}
		\int dy \, \Delta^n\left(P_1(x_1) \otimes \cdots \otimes  P_n(x_n);Q(y)\right)_0 = 0 \; .
	\end{equation}
\end{proposition}

\begin{proof}
	Let $(x_1,\dots, x_n) \in \mathbb{M}^n$ be given and let $\mathscr{O} = (x + V_+) \cap (y + V_-)$ for some $x,y \in \mathbb{M}$ be an open double cone that contains all $x_1,\dots, x_n$. Let $g \in \mathscr{D}(\mathbb{M})$ be a test function s.t. $g(x)=1 \; \forall \; x \in U$, where $U$ is a neighbourhood of $\overline{\mathscr{O}}\subset\mathbb{M}$.
	From equation (\ref{eq:locanom})  we know that $\Delta$ is local, that is $\Delta^n\left(P_1(x_1) \otimes\dots\otimes  P_n(x_n);Q(y)\right)$ as a distribution in any of the $x_i$ is supported only at the point $x_i = y$. Using this and expression (\ref{eq:start}) for the anomaly map, we may insert $g$ as above into the integral:
	\begin{align}
		\label{eq:sonnenschein}
		\frac{\hbar^n}{i^n}
		\int dy \; \Delta^n&\big(P_1(x_1) \otimes\dots\otimes  P_n(x_n);Q(y)\big)_0 \notag\\
		= &\frac{\hbar^n}{i^n}
		\int dy \, g(y)\, \Delta^n\big(P_1(x_1) \otimes\dots\otimes  P_n(x_n);Q(y)\big)_0  \notag\\
		= &-\int dy \, g(y)\,\partial_{\mu}^y \,T_{n+1}\big( P_1(x_1)\otimes\cdots\otimes P_n(x_n)\otimes j^{\mu}(y) \big)_0 \notag\\
		&+\hbar\sum_{l=1}^{n} 
		T_{n}\big( P_1(x_1)\otimes\cdots\otimes(\theta P_l)(x_l)\otimes\dots\otimes P_n(x_n) \big)_0 \notag\\
		&+ \hbar \sum_{l=1}^{n} \big(\partial^{\mu}g(x_l)\big)\cdot
		T_{n}\big( P_1(x_1)\otimes\cdots\otimes(\theta_{\mu} P_l)(x_l)\otimes\dots\otimes P_n(x_n) \big)_0
	\end{align}
	In the second last line, the $\delta$-distribution sets $g(x_l) =1$. In the last line, we use integration by parts to put the derivative onto  $g$, so this term vanishes since $g$ is constant on all $x_l$. The last two lines in equation (\ref{eq:start}) vanish since we are restricting to on-shell fields. \\
	
	Now consider $\partial^{\mu}g: \mathbb{M} \rightarrow \mathbb{M}$. Since this map is constant everywhere in $\mathscr{O}$, we have $(\text{supp}\, \partial^{\mu}g) \;  \cap\; \mathscr{O} = \emptyset$. Hence we may decompose it into $\partial^{\mu}g = a^{\mu} - b^{\mu}$ such that $\text{supp} \, a^{\mu} \cap (\mathcal{O} + \overline{V}_-) = \emptyset$ and  $\text{supp} \, b^{\mu} \cap (\mathcal{O} + \overline{V}_+) = \emptyset$. Then by causal factorization (\ref{eq:causality}) and using the AWI (\ref{eq:AWI}) the $T$-product in the third line of (\ref{eq:sonnenschein}) becomes
	\begin{align}
		\label{eq:flaco}
		T&\big( P_1(x_1)\otimes\cdots\otimes P_n(x_n)\otimes j^{\mu}(\partial_{\mu}g) \big)_0 \notag\\
		&= j^{\mu}(a_{\mu})_0 \star T\big(P_1(x_1) \otimes\dots\otimes  P_n(x_n)\big)_0 
		- T\big(P_1(x_1) \otimes\dots\otimes  P_n(x_n)\big)_0 \star j^{\mu}(b_{\mu})_ 0 \notag\\
		&= [j^{\mu}(a_{\mu})_0, T\big(P_1(x_1) \otimes\dots\otimes  P_n(x_n)\big)_0]
		+ T\big(P_1(x_1) \otimes\dots\otimes  P_n(x_n)\big)_0 \star j^{\mu}(\partial_{\mu}g)_ 0 \; ,
	\end{align}
	where for the second identity we have just added $T(\dots)_0\star j^{\mu}(a_{\mu}-a_{\mu})_0$. The second term in the last line vanishes since the on-shell free current is conserved. So the $T$-product on the l.h.s equals the commutator in the last line.\\
	
	\begin{figure}[h!]
		\centering
		\includegraphics[width=1\textwidth]{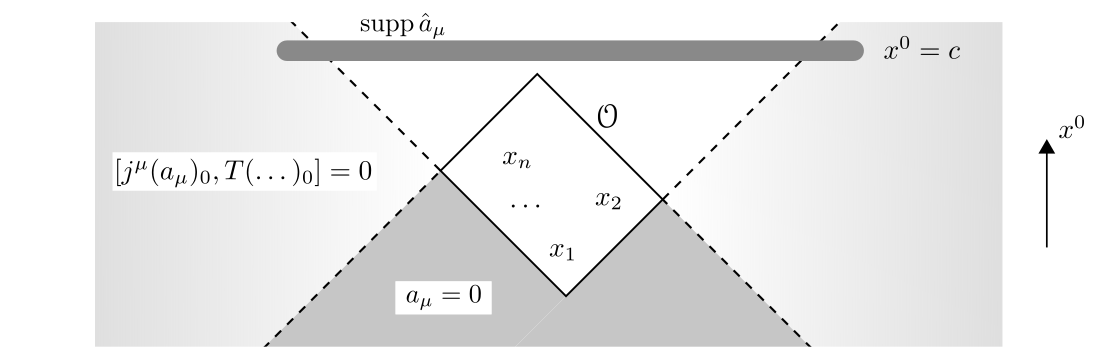}
		\caption[Testbild]{The double cone $\mathscr{O}$ contains the points $x_1,\dots,x_n$ surrounded by the regions where $a_{\mu}$ vanishes and where it can be varied without affecting the commutator. The new choice $\hat{a}_{\mu}$ which is supported only outside $\mathscr{O}$.\label{fig:dcone}}
	\end{figure}
	From (\ref{eq:tsupp}) we know that $\text{supp} \,  T\big(P_1(x_1) \otimes \cdots \otimes  P_n(x_n)\big)_0 \subset \mathscr{O}$. So due to spacelike commutativity (proposition \ref{prop:spacecomm}) we may vary the values of $a^{\mu}(z)$ at points $z$ that are spacelike separated from $\mathscr{O}$ without affecting the commutator $[j^{\mu}(a_{\mu})_0, T(\dots)_0]$ in (\ref{eq:flaco}). So -- by choosing a different $\hat{g}$ still satisfying the same property -- we may take $\hat{a}_{\mu}$ to have only a $x^0$ component and to be supported at an approximation of a time slice $x^0 = c$ of spacetime outside $\mathscr{O}$ (see figure \ref{fig:dcone}). That is
	\begin{equation*}
		\hat{a}_{\mu}(x) =- \delta_{\mu 0} h(x^0) \quad \text{where} \int d x^0 \, h(x^0) = 1, \quad h \in \mathscr{D} ([c - \varepsilon, c + \varepsilon])\, , \quad c \in \mathbb{R} \; 
	\end{equation*}
	for some $\varepsilon > 0$, where $c$ is large enough to be outside $\mathcal{O}$ and $\delta_{\mu 0}$ is the Kronecker delta. Due to proposition \ref{prop:Q_0} we know that the integral 
	\begin{equation*}
		[Q_0,T(\dots)_0] := \int d\vec{x}\, [j^0(x^0,\vec{x})_0,T(\dots)_0]
	\end{equation*}
	exists and does not depend on $x^0$ and that $[Q_0,P(x)_0]= \hbar\,(\theta P)(x)_0$ for $P \in \mathscr{P}^{(1)}$. Since we have $P_1,\dots,P_n \in \mathscr{P}_{(\phi^*\phi)^2}\subset\mathscr{P}^{(1)}$ we know that the fields contained in $T\big(P_1(x_1) \otimes\dots\otimes  P_n(x_n)\big)_0$ are also elements of $\mathscr{P}^{(1)}$. 
	Using the chosen $\hat{a}^{\mu}$ we compute the commutator in (\ref{eq:flaco}) as 
	\begin{align*}
		[j^{\mu}(\hat{a}_{\mu})_0, T(\dots)_0]_{\star}
		&= -\int d x^0 \, h(x^0) \int d \vec{x} \, [j^0(x^0,\vec{x})_0, T(\dots)_0]_{\star} \notag\\
		&= -[Q_0, T(\dots)_0]_{\star} = -\hbar\,\big(\theta\, T(\dots)\big)_0 \; .
	\end{align*}
	Putting this back into (\ref{eq:sonnenschein}) yields 
	\begin{align}
		\label{eq:strawberries}
		\frac{\hbar^n}{i^n}
		\int dy \, \Delta^n\left(P_1(x_1) \otimes\dots\otimes  P_n(x_n);Q(y)\right)_0 
		= -\hbar\, \theta\, T_{n}\big( P_1(x_1)\otimes\cdots\otimes P_n(x_n) \big)_0 \notag\\
		+\,\hbar\sum_{l=1}^{n} 
		T_{n}\big( P_1(x_1)\otimes\cdots\otimes(\theta P_l)(x_l)\otimes\dots\otimes P_n(x_n) \big)_0 \; .
	\end{align}
	The r.h.s. vanishes due to charge number conservation (\ref{eq:chconsv}). 
\end{proof}

\begin{remark}
	Equation \ref{eq:strawberries} shows that charge number conservation is a necessary condition for the MWI to be satisfied, that is for $\Delta^n=0$ to hold true. 
\end{remark}

\subsection{Structure of the anomalous term}

\label{sec:4struc}

Here we express the anomalous term via its causal Wick expansion and apply the anomalous MWI from theorem \ref{thm:aMWI} to derive further properties of the involved VEVs.

\begin{proposition}
	To prove theorem \normalfont{\ref{thm:bachelor}}\textit{ it suffices to show the statement for all va-cuum expectation values $d(P_1,\dots,P_n):= \omega_0 \circ \Delta (P_1,\dots,P_n;Q)$ with arguments $P_1,\dots,P_n$ that are at least quadratic in the basic fields.}
\end{proposition}

\begin{proof}
	From theorem \ref{thm:aMWI} we know that $\Delta^n$ satisfies the axiom field independence, hence we may express it in terms of the causal Wick expansion (see theorem \ref{thm:cWick})
	\begin{align}
		\label{eq:anWick}
		\Delta^n&\left(P_1(x_1) \otimes\dots\otimes  P_n(x_n);Q(y)\right)\notag\\
		&= \sum_{\underline{P}_l\subseteq P_l} d\big( \underline{P}_1,\dots,\underline{P}_n \big)(x_1-y,\dots,x_n-y)\,
		\overline{P}_1(x_1)\cdots\overline{P}_n(x_n) \; ,
	\end{align}
	where we have used that the only submonomial of $Q$ and $Q^*$ is $c\in\mathbb{C}$ and its contribution vanishes, see equation (\ref{eq:anomoffs}). Due to proposition \ref{prop:argums} all VEVs in the sum that have some $\underline{P}_j$ as argument containing only a basic field also vanish. This shows the proposition. 
\end{proof}

\begin{proposition}
	We can write $d(P_1,\dots,P_n)$ as 
	\begin{align}
		\label{eq:tatue}
		&d\big(P_1,\dots,P_n\big)(x_1-y,\dots,x_n-y) = \partial_{\mu}^y
		u^{\mu}(P_1,\dots,P_n)(x_1-y,\dots,x_n-y)\;,  \notag\\
		\text{\normalfont{where}}\quad& u^{\mu} = \sum_{|a|\leq \omega(P_1,\dots,P_n)-1} C_a^{\mu} \big(P_1,\dots,P_n\big)\,\partial^a \delta(x_1-y,\dots,x_n-y)\;, \notag\\
		\text{\normalfont{with}}\quad &a\in(\mathbb{N}^d)^n \,, \quad \omega(P_1,\dots,P_n):=\sum_{i=1}^{n}\text{\normalfont{dim}}\,P_j+4-4n \; ,\notag\\
		\text{\normalfont{for some}}\quad & C_a^{\mu} (P_1,\dots,P_n) \in \mathbb{C}\; .
	\end{align}
\end{proposition}

\begin{proof}
	By equation (\ref{eq:locanom}) in theorem \ref{thm:aMWI}, the anomaly maps $\Delta^n$ are local fields. This carries over to their VEVs, which we may write as 
	\begin{equation}
		\label{eq:d_exp}
		d\big(P_1,\dots,P_n\big)(x_1-y,\dots,x_n-y) 
		=  \sum_{a} \tilde{C}_a \big(P_1,\dots,P_n\big)\,\partial^a \delta(x_1-y,\dots,x_n-y) \; ,
	\end{equation}
	for some $\tilde{C}_a (P_1,\dots,P_n) \in \mathbb{C}$. Furthermore by using equation (\ref{eq:a_restr}) in the same theorem with four space time dimensions and $\text{dim}\,Q = 1$, $\dim \tilde{C}_a = 0$ we conclude that the sum over $a$ is restricted to 
	\begin{equation}
		\label{eq:anom_restr}
		|a| \leq \sum_{i=1}^{n}\text{\normalfont{dim}}\,P_j+4-4n =: \omega(P_1,\dots,P_n) \; .
	\end{equation}
	We have shown equation (\ref{eq:int0}) for the on-shell anomaly, hence this holds true also for $d$ (since it contains no fields)
	\begin{equation*}
		\int dy \; d\big(P_1,\dots,P_n\big)(x_1-y,\dots,x_n-y) = 0\; .
	\end{equation*}
	By equation (\ref{eq:poinc}) we may therefore write it as the total divergence $d(\dots) = \partial_{\mu}^y u^{\mu}(\dots,y)$ for some $u^{\mu}$ with a Lorentz index $\mu$ that is again local. Since now $u^{\mu}$ must contain one derivative less than $d$, we arrive at the expression in (\ref{eq:tatue}) for $u^{\mu}$. 
\end{proof}

\begin{lemma}
	 \label{lemm:poinc}
	 The VEVs  $d\big(P_1,\dots,P_n\big)(x_1-y,\dots,x_n-y)$ are Lorentz covariant and satisfy the $*$-structure condition, hence the $u^{\mu}$ can be chosen to do so too. 
\end{lemma}

\begin{proof}
	Lorentz covariance can be seen from equation (\ref{eq:start}) or (\ref{eq:anom_lorz}), the statement about the $*$-structure follows from equation (\ref{eq:anomstar}). 
\end{proof}

\subsection{Admissible finite renormalizations}

\label{sec:5renom}

\begin{proposition}
	\label{lemm:admren}
	Consider the terms $t( P_1,\dots,P_n,j^{\mu})$ appearing in the causal Wick expansion of the first term on the r.h.s of equation \normalfont{(\ref{eq:start})}\textit{ and  a  $u^{\mu}$ as in equation }\normalfont{(\ref{eq:tatue})}\textit{. Then the renormalized expressions
	\begin{equation}
		\label{eq:renm}
		t\big( P_1,\dots,P_n,j^{\mu} \big) \rightarrow
		t\big( P_1,\dots,P_n,j^{\mu} \big)
		+ \frac{\hbar^n}{i^n}\, u^{\mu}\big( P_1,\dots,P_n\big)
	\end{equation}
	satisfy all basic axioms and renormalization conditions except for the axiom} \normalfont{(iii)}\textit{ on symmetry in definition} \normalfont{\ref{def:Tn}}\textit{. The anomaly in equation} \normalfont{(\ref{eq:start})} \textit{vanishes by performing this renormalization, hence} \normalfont{(\ref{eq:compMWI})} \textit{holds true.}
\end{proposition}

\begin{proof}
	Equation (\ref{eq:renm}) is a finite renormalization in the sense of definition \ref{def:finren} due to the particular form of $u^{\mu}$ in equation (\ref{eq:tatue}). That the basic axioms and renormalization conditions remain satisfied may be checked one by one (see \cite[chap. 5.2.2]{duetsch19}). We point out that Poincaré covariance and $*$-structure follow from lemma \ref{lemm:poinc} and scaling degree from the restriction on $a$ in equation (\ref{eq:tatue}). By putting the renormalized $t$ into equation (\ref{eq:start}) we see that -- after applying the derivative on the r.h.s -- the terms $d= \partial^y_{\mu} u^{\mu}$ cancel out the anomalies in each order of the causal Wick expansion, hence the MWI now holds true. 
\end{proof}

\subsection{Case distinction}
\label{sec:6case}

Proposition \ref{lemm:admren} leaves one possibility open: If $t( P_1(x_1),\dots,P_n(x_n),j^{\mu}(y))$ has further symmetries than the ones discussed so far, it is not yet clear whether we can find a choice of $u^{\mu}$ that has the same symmetries. The situation now is as follows. For a particular choice of renormalization of all the $T$-products, the anomaly maps $\Delta^n$ are given for all $n$ by equation (\ref{eq:sonnenschein}). By Wick expanding $\Delta^n$ as in (\ref{eq:anWick}) the maps $d(\dots)$ are uniquely fixed, and due to the uniqueness of (\ref{eq:locanom}), the expansion coefficients $\tilde{C}_a$ of $d(\dots)$ in (\ref{eq:d_exp}) are also determined -- but we don't know their particular values. Now $u^{\mu}$ is not unique. All we know is that some $u^{\mu}$ of the form in equation (\ref{eq:tatue}) exists satisfying $\partial_{\mu}^y u^{\mu} =d$. So there is some arbitrariness in the choice of the ${C}^{\mu}_a$ in (\ref{eq:tatue}) namely we can always add a term $\hat{u}^{\mu}$ satisfying $\partial_{\mu}^y \hat{u}^{\mu} = 0$. So the task is this: For given $P_1,...,P_n$ we need to find a $u^{\mu}$ (that is find coefficients ${C}^{\mu}_a$) with the same symmetries as $t(P_1,\dots,P_n,j^{\mu})$ such that $\partial_{\mu}^y u^{\mu} =d(P_1,...,P_n)$, where the form of $d(P_1,...,P_n)$ is given by (\ref{eq:d_exp}) but we don't know the particular values of the coefficients $\tilde{C}_a$. \\

The additional symmetries of $t(P_1,\dots,P_n,j^{\mu})$ we have not discussed yet are due to some of its arguments being equal. If $P_k = P_l$ for some $k,l\leq n$, the $t$-product will be symmetric in the arguments $x_k \leftrightarrow x_l$. Since $d(\dots)$ is symmetric in its arguments, this symmetry will carry over to the anomaly. Now in $d = \partial_{\mu}^y u^{\mu}$ there is only a derivative w.r.t $y$, so for any choice $u^{\mu}$ will satisfy symmetries of this kind. The same argument applies if $P_k = P_l$ involve any Lorentz indices (like $P^{\mu}_k = \partial^{\mu}\phi$). A different kind of symmetries occurs if there are factors of $j^{\nu_l}(x_l)$ among the $P_1,\dots,P_n$. Then $t(\dots)$ will be symmetric with respect to $(y,\mu)\leftrightarrow (x_l,\nu_l)$ for all $l$. In this case it is not clear whether a $u^{\mu}$ having this symmetries can always be found. These remaining cases are addressed by the following proposition. 

\begin{proposition}
	\label{prop:last}
	Let $n\in \mathbb{N}$ and $m\leq n$ be arbitrary. Consider polynomials $P_1,\dots,P_n$ where $P_l(x_l) = j^{\nu_l}(x_l)$ for $l\leq m$. Then we may find a $u^{\mu}$ as in \normalfont{(\ref{eq:tatue})} \textit{that is totally symmetric in $(y,\mu)\leftrightarrow (x_l,\nu_l)$ for $l\leq m$.} 
\end{proposition}

	
\begin{proof}
	
We will prove the statement by a case distinction, that is by finding all possible combinations of $P_1,\dots,P_n$ containing some $j^{\nu_1},\dots j^{\nu_m}$ with $m \leq n$ that lead to a non-zero anomaly $d^{\nu_1\dots\nu_m}$, and then work out explicitly a suitable $u^{\nu_1\dots\nu_m\mu}$ with the required symmetries.   \\

Consider equation (\ref{eq:d_exp}) for the anomaly $d$. For a choice of $P_1,\dots,P_n$, only terms with $|a| \leq \omega(P_1,\dots,P_n)$ contribute to $d$. So the higher the sum of the mass dimensions of $P_1,\dots,P_n$, the more contributions to $d$ we get. We have the following mass dimensions
\begin{equation*}
	\text{dim}\, \phi = \text{dim}\, \phi^* = 1 \; ,\quad
	\text{dim}\,j^{\mu} = 3\; ,\quad \; \text{dim}\, L_{\text{int}} = 4 \; .
\end{equation*}
The polynomial with the highest mass dimension in $\mathscr{P}_0$ is the interaction $L = (\phi^* \phi)^2$. We consider choices of $P_1,...,P_n$ where at least one $P_i$ equals $j^{\mu}$ and start with the choice of $P_1,\dots,P_n$ with the highest mass dimension -- that is all $P_l$ equal $L$ except for one $j^{\mu}$. Then we reduce the mass dimension by dropping factors of $\phi$ and $\phi^*$ from some of the $P_1,\dots,P_n$, until we get no more contributions to $d$, that is until $\omega(P_1,\dots,P_n)\leq 0$. To rule out particular cases of $P_1,\dots,P_n$ we will use Furry's theorem \ref{thm:gfurry} and charge number conservation (the contraposition of lemma \ref{lemm:pups}). This leads to all the possibilities listed in table \ref{tab:0}. 
	
\begin{table}[h!]
		\centering
		\renewcommand{\arraystretch}{1.3}	
		\begin{tabular}{c|l|c|cc}
	\#&\multicolumn{1}{c|}{$P_1,\dots,P_n$} 
		& $\omega(P_1,\dots,P_n)$ & $d$ vanishes due to&\\\cline{1-4}
			
	1&$\underbrace{L,\dots,L}_{n-1},j^{\nu}$ 
		& 3			&  					& \textbf{case I}\\
	2&$\underbrace{L,\dots,L}_{n-2},j^{\nu_1},j^{\nu_2}$ 
		&  			& Furry 			&\\
	3&$\underbrace{L,\dots,L}_{n-3},j^{\nu_1},j^{\nu_2},j^{\nu_3}$ 
		& 1			&  					& \textbf{case II} \\
	4&$\underbrace{L,\dots,L}_{n-4},j^{\nu_1},\dots j^{\nu_4}$ 
		& 0 		& 					&\\
	5&$\underbrace{L,\dots,L}_{n-2},\phi^* \phi^2,j^{\nu}$ 
		&  			& charge number 	&\\
	6&$\underbrace{L,\dots,L}_{n-2},{\phi^*}^2\phi,j^{\nu}$ 
		&  			& charge number 	&\\
	7&$\underbrace{L,\dots,L}_{n-2},\phi^* \phi,j^{\nu}$ 
		& 1 		&  					& \textbf{case III} \\
	8&$\underbrace{L,\dots,L}_{n-3},\phi^2 \phi^*,\phi \phi^*,j^{\nu}$ 
		& 0			& 					&\\
	&\multicolumn{1}{c|}{\vdots}
		& $\leq$ 0 	& 					& 

		\end{tabular}
		\caption{Several choices of $P_1,\dots,P_n$ are considered that may lead to non zero contributions to the anomaly $d$. Only the three cases I-III need to be further discussed.}
	\label{tab:0}
\end{table}

\begin{lemma}
	The contribution in line \normalfont{2}\textit{ of table }\normalfont{\ref{tab:0}} \textit{to the anomaly $d$ vanishes due to Furry's theorem. }
\end{lemma}

\begin{proof}
	Taking the VEV on both sides of \ref{eq:start} yields
	\begin{align}
	\label{eq:MWI_VEV}
	\frac{\hbar^n}{i^n}
	d&\left(P_1\dots P_n\right) 
	= 
	- \partial_{\mu}^y \, t\big( P_1\dots P_n, j^{\mu}(y) \big) \notag\\
	&+\hbar\sum_{l=1}^{n}\delta(y-x_l)\, 
	t\big( P_1 ,\dots,(\theta P_l)(x_l),\dots,P_n \big)
	\notag\\
	& -\hbar\,\partial_y^{\mu} \Big(
	\sum_{l=1}^{n}\delta(y-x_l)\,
	t\big( P_1,\dots,(\theta_{\mu} P_l)(x_l),\dots, P_n(x_n) \big)\Big)  \; ,
	\end{align}
	where we have omitted most of the arguments. We calculate 
	\begin{equation}
		\label{eq:thetas}
		\theta L = 0 \;,\quad \theta j^{\nu} = 0 \;,\quad \theta_{\mu} L = 0 \;,\quad
		\theta_{\mu} j^{\nu} = -2 i\, \delta_{\mu}^{\nu}\, \phi^* \phi \; .
	\end{equation}
	Now putting in $P_1,\dots,P_n$ as in line 2 of table \ref{tab:0} and using these relations leaves on the r.h.s. only VEVs with an odd number of currents $j^{\nu_l}$ and all other arguments either $L$ or $\phi^* \phi$. Both $\beta_C(L) = L$ and $\beta_C(\phi^* \phi) = \phi^* \phi$ are even under charge number conjugation, hence we may apply Furry's theorem in its general form of equation (\ref{eq:gfurry}) to show that all contributions to $d$ vanish. 
\end{proof}

\begin{lemma}
	The contributions in lines \normalfont{5}\textit{ and }\normalfont{6}\textit{ of table }\normalfont{\ref{tab:0}}\textit{ to the anomaly $d$ vanish due to charge number conservation.  }
\end{lemma}

\begin{proof}
	The $P_1,\dots P_n$ in these lines of the table include one term like $P_j= \phi^2 \phi^*$ for which $a_j - b_j \neq 0$ and all other terms $P_i$ have $a_i-b_i = 0$. So for each of the $P_1,\dots P_n$ we have $\sum_{j=1}^{n} (a_j-b_j) \neq 0$. Due to equations (\ref{eq:thetas}) and the additional relations 
	\begin{equation*}
		\theta_{\mu} \phi^2 \phi^* = \theta_{\mu} \phi (\phi^*)^2 = 0 \; ,
	\end{equation*}
	this statements holds true for all $t$-products appearing on the r.h.s of equation (\ref{eq:MWI_VEV}). By charge number conservation in (the contraposition of) lemma \ref{lemm:pups} these $t$-products vanish, hence there is no contribution to $d$. 
\end{proof}

We are left with the three cases I-III in table \ref{tab:0} and now want to explicitly find suitable renormalizations $u^{\mu}$ for each of them.\footnote{We use $u^{\mu}$ or simply $u$ to refer to a general renormalization, although the specific $u$ appearing in the following will in general have more than one Lorentz index.} We begin with the latter two and then turn to the most involved first case. \\

\textbf{Case II.} Writing down equation (\ref{eq:MWI_VEV}) for case II with $d = \partial_{\mu}^y u^{\mu}$ and labeling the arguments as follows yields (we absorb the factors of $\hbar$ and $i$ into the constants and omit the arguments of the $C_a$)
\begin{align}
	\label{eq:case2}
	&- \partial_{\mu}^y \, t\big(\overbrace{L,\dots,
	L}^{m:=n-3},j^{\nu_1},j^{\nu_2},j^{\nu_3},j^{\mu} \big) (x_{11}-y,\dots,x_{1m}-y,x_{21}-y,\dots,x_{23}-y) \notag\\
	&+2i\hbar\,\partial_y^{\mu} \Bigg( \sum_{l=1}^{3}\delta(y-x_{2l})\,\delta_{\mu}^{\nu_l}
	t\big(L,\dots,L,j^{\nu_1},\phi^* \phi,j^{\nu_3}\big) (x_{11}-x_{23},\dots)\Bigg)  \; , \notag\\
	&\qquad= \partial_{\mu}^y  C_{\text{II}}^{\mu\nu_1\nu_2\nu_3} \, 
	\prod_{l,j=1}\delta(x_{lj}-y) \; .
\end{align}
Here $\omega =1$, hence $|a|= 0$ and there are no derivatives contained in $u^{\mu}$. We know by lemma \ref{lemm:poinc} that $C_{\text{II}}^{\mu\nu_1\nu_2\nu_3}$ could be any constant Lorentz invariant tensor of rank $4$. We write down its most general form as
\begin{equation*}
	C_{\text{II}}^{\mu\nu_1\nu_2\nu_3} = C_{\text{II}}^1\cdot g^{\mu\nu_1}g^{\nu_2\nu_3} + C_{\text{II}}^2\cdot g^{\mu\nu_2}g^{\nu_1\nu_3} + C_{\text{II}}^3\cdot g^{\mu\nu_3}g^{\nu_1\nu_2} \; , \qquad C_{\text{II}}^k \in \mathbb{C} \; ,
\end{equation*}
where we have used that any constant Lorentz invariant tensor (or tensor density) is composed of the metric $g^{\mu\nu}$ and the totally antisymmetric Levi-Civita symbol $\epsilon^{\alpha\beta\gamma\delta}$ \linebreak \cite[chap. 2]{may18}. In equation (\ref{eq:case2}) we see that the l.h.s is symmetric under permutations of $({\nu}_1,x_{21}) \leftrightarrow({\nu}_2,x_{22})\leftrightarrow ({\nu}_3,x_{23})$, hence the right hand side must be so too. The product of $\delta$-distributions is symmetric in all its arguments and the derivative w.r.t to $y$ is not affected by this permutation, so $C_{\text{II}}^{\mu\nu_1\nu_2\nu_3}$ must be symmetric under permutation of all its Lorentz indices. This lets us rule out the Levi-Civita symbol and conclude that $C_{\text{II}}^1 = C_{\text{II}}^2 = C_{\text{II}}^3 := C_{\text{II}}$. Hence any possibly occurring $u_{\text{II}}$ is of the form 
\begin{equation*}
	u^{\mu\nu_1\nu_2\nu_3}_{\text{II}} = C_{\text{II}} \Big(g^{\mu\nu_1}g^{\nu_2\nu_3} + g^{\mu\nu_2}g^{\nu_1\nu_3} +  g^{\mu\nu_3}g^{\nu_1\nu_2}\Big)\cdot\prod_{l,j=1}\delta(x_{lj}-y) \; .
\end{equation*}
Now this expression is invariant under the required permutation symmetry in the arguments $({\nu}_1,x_{21}) \leftrightarrow({\nu}_2,x_{22})\leftrightarrow ({\nu}_3,x_{23})\leftrightarrow(\mu,y)$ of the $t$-product, hence it is an admissible renormalization. \\

\textbf{Case III.} In this case, equation (\ref{eq:MWI_VEV}) yields the following expression for the anomaly, where $P_{n-1}(z)=\phi^* \phi(z)$: 
\begin{align}
	\label{eq:case3}
	&- \partial_{\mu}^y \, t\big(\overbrace{L,\dots,
	L}^{m:=n-2}\phi^*\phi,j^{\nu},j^{\mu} \big) (x_{11}-y,\dots,x_{1m}-y,z-y,x_2-y)\notag\\
	&+2i\hbar\,\partial_y^{\nu} \Big(\delta(y-x_2)\,
	t\big(L,\dots,L,\phi^* \phi,\phi^* \phi\big) (x_{11}-x_2,\dots,x_{1m}-x_2, z-x_2)\Big) \notag\\
	&= \partial_{\mu}^y\,  C_{\text{III}}^{\mu\nu} \, \delta(x_{11}-y,\dots,x_{1m}-y,z-y,x_2-y)
\end{align}
Again $|a|=0$ and $u^{\mu}$ contains no derivatives. In this case, $C_{\text{III}}^{\mu\nu}$ is given by the most general Lorentz invariant tensor of rank $2$, which is simply the metric. We get 
\begin{equation*}
	u^{\mu\nu}_{\text{III}} = C_{\text{III}}\, g^{\mu\nu}\cdot \delta(x_{11}-y,\dots,x_{1m}-y,z-y,x_2-y) \; ,\qquad C_{\text{III}} \in \mathbb{C} \; .
\end{equation*}
This is symmetric in $({\nu},x_2)\leftrightarrow(\mu,y)$ and yields an admissible renormalization.  \\

\textbf{Case I.} From equation (\ref{eq:MWI_VEV}) we get

\begin{align}	
	\label{eq:case1}
	&- \partial_{\mu}^y \, t\big(\overbrace{L,\dots,
	L}^{m:=n-1},j^{\nu},j^{\mu} \big) (x_{11}-y,\dots,x_{1m}-y,x_2-y)  \notag\\
	&+2i\hbar\,\partial_y^{\nu} \Big(\delta(y-x_2)\,
	t\big(L,\dots,L,\phi^* \phi\big) (x_{11}-x_2,\dots,x_{1m}-x_2)\Big)  \; , \notag\\
	&\qquad= \partial_{\mu}^y \sum_{|a|\leq 2} C_{\text{I},a}^{\mu\nu} \,\partial^a \delta(x_{11}-y,\dots,x_{1m}-y,x_2-y) \; .
\end{align}
Since the VEV $d$ of the anomaly depends only on relative coordinates, one of the variables $x_{11},\dots,x_{1m},x_2,y$ is dependent. Using the chain rule we can express all derivatives w.r.t $y$ as 
\begin{equation*}
	\partial^{\mu}_y = - \Big(\partial^{\mu}_{2} + \sum_{i} \partial^{\mu}_{i} \Big) \; , \qquad \text{where}\; i \in \{11,...,1m\} \; , \quad \partial_{k} \equiv \partial_{x_k} \; ,
\end{equation*}
and eliminate them from equation (\ref{eq:case1}). The expression for $u_{\text{I}}$ must be a Lorentz tensor of rank $2$ containing at most two derivatives. There is no way to write down such a tensor containing only one derivative. The contributions with two derivatives can be either contracted with each other or have both free indices. Terms involving $\epsilon^{\alpha\beta\gamma\delta}$ will be ruled out later due to their antisymmetry. So we write down the most general form of $u_{\text{I}}$ as 
\begin{align}
	\label{eq:hio}
	u_{\text{I}}^{\mu\nu} 
	= \Big( g^{\mu\nu}&\sum_{i,j} a_{ij}\, \partial_i^{\alpha}\partial_{j\alpha} +  
	\sum_{i,j} b_{ij}\, \partial_i^{\mu}\partial_j^{\nu}
	+ g^{\mu\nu}c_0 \Big)\,
	\delta(x_{11}-y,\dots,x_{1m}-y,x_2-y) \; , \notag \\
	\text{with} \qquad 
	&a_{ij}, b_{ij}, c_0 \in \mathbb{C} \; ,\qquad 
	i,j \in \{11,...,1m,x_2\} \; .
\end{align}
These contributions can be further restricted by symmetries they have to satisfy. The l.h.s. of equation (\ref{eq:case1}) is symmetric under the permutation of any of the $x_{11},\dots,x_{1m}$ and the derivative $\partial_{\mu}^{y}$ does not affect this, so $u_{\text{I}}^{\mu\nu}$ does have this symmetry too. We will now write down all possible contributions with two derivatives to $u_{\text{I}}^{\mu\nu}$ satisfying this symmetry by distinguishing the cases of the two derivatives acting (i) both on the same $x_{1i}$ (ii) on $x_{1i}$ and $x_{1j}$ for $i\neq j$ (iii) on some $x_{1i}$, and on $x_2$, (iv) both on $x_{2}$ as follows,
\begin{align}
	\label{eq:basis1}
	&\text{(i)} && g^{\mu\nu}\sum_{i} {\dalemb}_{i}
	&&\sum_{k} \partial_k^{\mu}\partial_k^{\nu}
	\notag\\
	&\text{(ii)} && g^{\mu\nu}\sum_{i\neq j} \partial_i^{\alpha}\partial_{j\alpha}
	&&\sum_{k\neq l} \partial_k^{\mu}\partial_l^{\nu}
	\notag\\
	&\text{(iii)} && g^{\mu\nu}\partial_{2\alpha}\sum_{i} \partial_{i}^{\alpha}
	&& \partial_{2}^{\mu}\sum_{k}\partial_k^{\nu}
	&& \partial_{2}^{\nu}\sum_{k}\partial_k^{\mu}
	\notag\\
	&\text{(iv)}  && g^{\mu\nu}\partial_{2}^{\alpha} \partial_{2\alpha}
	&&\partial_{2}^{\mu} \partial_{2}^{\nu} \; ,
\end{align}
where now $i,j,k,l \in \{11,...,1m\}$. The first and second columns correspond to the first and the second term in equation (\ref{eq:hio}), for case (iii) there are two contributions from the second term. This covers all possible cases, and since all involved variables are independent, the $9$ obtained expressions are easily seen to be linearly independent. So we may say that these $9$ objects -- each one multiplied by the $\delta$-distribution in (\ref{eq:hio}) -- form a basis for the vector space of all possible $u_{\text{I}}^{\mu\nu}$ with two derivatives that are symmetric in the $x_{11},\dots,x_{1m}$. We now give a different set of $9$ terms arranged in three groups (1)-(3) that are better suited for the following computations.
\begin{align}
	\label{eq:basis2}
	&\text{(1)}
	&& g^{\mu\nu}\sum_{i} {\dalemb}_{i}
	&&\sum_{i} \partial_i^{\mu}\partial_i^{\nu}
\notag\\
&\text{(2)}
	&& \partial_{2}^{\mu} \partial_{y}^{\nu}
	&& \partial_{y}^{\mu} \partial_{2}^{\nu}
	&& g^{\mu\nu} \partial_{y}^{\alpha} \partial_{2\alpha}
\notag\\[0.9em]
	&\text{(3)}
	&& g^{\mu\nu} {\dalemb}_{2}
	&& g^{\mu\nu} {\dalemb}_{y}
	&& \partial_{2}^{\mu} \partial_{2}^{\nu}
	&& \partial_{y}^{\mu} \partial_{y}^{\nu} \; ,
\end{align}
where again $i \in \{11,...,1m\}$. By using
\begin{equation*}
	\sum_i \partial_i = -\partial_2-\partial_y  \qquad 
	\text{and} \qquad \sum_{i\neq j} \partial_i \partial_j = \Big(\sum_i \partial_i\Big)^2 - \sum_i \partial_i \partial_i \;,
\end{equation*}
(where we omit all Lorentz indices) one may express all elements of the old basis (\ref{eq:basis1}) as linear combinations of the new terms (\ref{eq:basis2}), hence they also form a basis of the same space: Every possible $u_{\text{I}}^{\mu\nu}$ with the mentioned symmetry is a linear combination of these terms. Each of the groups (1)-(3) in equation (\ref{eq:basis2}) transforms separately under the symmetry in $(y,\mu)\leftrightarrow(x_2,\nu)$ required from $u^{\mu}$, so they can be discussed independently. The terms in group (1) are invariant under this transformation, so every contribution from these terms to the anomaly will be an allowed renormalization for the $t$-product. The same holds true for all terms in group (2). So we only need to discuss the remaining four terms in group (3). 
\\

Now consider (\ref{eq:case1}) and apply $\partial_{\nu}^2$ to both sides of the equation. The l.h.s. then becomes symmetric under $x_2 \leftrightarrow y$, so the same must hold true for the r.h.s., that is $\partial_{\nu}^2\partial_{\mu}^y u_{\text{I}}^{\mu\nu}$ is symmetric in $x_2 \leftrightarrow y$.  We will use this as a condition to restrict further on the possible contributions to $u_{\text{I}}^{\mu\nu}$ from group (3), by requiring that the antisymmetric part of every linear combination of the terms in (3) must vanish after applying the derivatives $\partial_{\nu}^2\partial_{\mu}^y$. Applying them yields the four terms 
\begin{equation*}
	\lambda\cdot {\dalemb}_2\; , \qquad 
	\lambda\cdot {\dalemb}_y\; , \qquad
	\lambda\cdot {\dalemb}_2\; , \qquad 
	\lambda\cdot {\dalemb}_y\; , \qquad \text{where}\;
	\lambda := \partial_{y}^{\alpha}\partial_{2\alpha} \; .
\end{equation*}
The condition on the possible linear combinations of these terms for $C_1,\dots,C_4 \in \mathbb{C}$ reads 
\begin{align*}
	0 &\overset{!}{=} \lambda \, \Big( C_1 \cdot({\dalemb}_2 -{\dalemb}_y)
	+ C_2 \cdot({\dalemb}_y -{\dalemb}_2) + C_3 \cdot({\dalemb}_2 -{\dalemb}_y) + C_4 \cdot({\dalemb}_y -{\dalemb}_2) \Big) 
	\notag\\
	&= \lambda \,\Big( {\dalemb}_2 \cdot (C_1-C_2+C_3-C_4) + {\dalemb}_y \cdot (-C_1+C_2-C_3+C_4) \Big) \qquad \forall\, x_2,y
	\notag\\
	&\Rightarrow C_1 = C_2 - C_3 + C_4 \; ,
\end{align*}
where we have antisymmetrized in $x_2 \leftrightarrow y$. So after replacing $C_1$ by the other $C_i$ we get that the most general remaining anomaly will be a linear combination of the form 
\begin{align*}
	d^{\nu} (\dots) = \partial_{\mu}^y \Big( 
	C_2\cdot (g^{\mu\nu} {\dalemb}_{y} + g^{\mu\nu} {\dalemb}_{2}) &+ 
	C_3\cdot (\partial_{2}^{\mu} \partial_{2}^{\nu} - 
	g^{\mu\nu} {\dalemb}_{2})  \notag\\
	&+ C_4\cdot (\partial_{y}^{\mu} \partial_{y}^{\nu} +g^{\mu\nu} {\dalemb}_{2})\,
	\Big) \cdot \delta(\dots)\; .
\end{align*}
The term in $C_2(\dots)$ is already symmetric in $(y,\mu)\leftrightarrow(x_2,\nu)$. After applying the derivative $\partial^{y}_{\mu}$ outside the large bracket, the term $C_4(\dots)$ yields the same contribution to the anomaly as $C_2(\dots)$. So we can renormalize away anomalies coming from $C_4(\dots)$ by using $C_2(\dots)$, which is admissible. The contribution from $C_3(\dots)$ is not symmetric in $(y,\mu)\leftrightarrow(x_2,\nu)$. To symmetrize it, we would need to add a term proportional to 
\begin{equation*}
	\partial_{y}^{\mu} \partial_{y}^{\nu} - 
	g^{\mu\nu} {\dalemb}_{y} \; ,
\end{equation*}
which vanishes after applying the derivative $\partial^{y}_{\mu}$. Hence we can renormalize away terms of the form $C_3(\dots)$ by their symmetrized version. So all possible anomalies can be removed by admissible renormalizations respecting the symmetries of $t$. This completes the proof of proposition \ref{prop:last} and hence of theorem \ref{thm:bachelor}. 
\end{proof}

\newpage

	\section{Examples of particular Ward Identities}
	\label{sec:examples}
	
We have shown that the complex scalar MWI (\ref{eq:compMWI}) can be satisfied in the quantum theory. In this last section we give some particular examples of Ward Identities following from (\ref{eq:compMWI}), by choosing particular polynomials $P_1,...,P_n \in \mathscr{P}_{(\phi^*\phi)^2}$ and computing the corresponding unrenormalized $t$-products on $\mathscr{D}(\check{\mathbb{M}}^n)$ by equation (\ref{eq:unrenT}). The statement of this thesis is then that these expressions can be renormalized satisfying all conditions in a way that the identity remains true. \\

We will make use of Feynman diagrams to represent the results as follows:  
\begin{align}
	\omega_0 \big(T(\phi^*(x_1),\phi(x_2))\big) \quad
	&= \hbar\cdot \Delta_F(x_1-x_2) \quad\;=:\quad \hbar \; \cdot \;
	\parbox[h][0.002\linewidth][c]{0.2\linewidth}{
		\includegraphics[width=\linewidth]{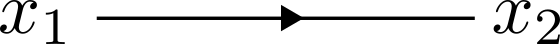}
	} \notag\\
	\omega_0 \big(T(\partial^{\mu}\phi^*(x_1),\phi(x_2))\big) 
	&= \hbar\cdot\partial_{x_1}^{\mu}\Delta_F(x_1-x_2) =:\quad \hbar \; \cdot \;
	\parbox[h][0.002\linewidth][c]{0.2\linewidth}{
		\includegraphics[width=\linewidth]{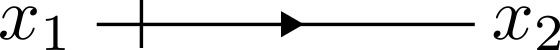}
	}
\end{align}

Exchanging the arguments of $\phi$ and $\phi^*$ in a contraction leaves the analytic expression invariant, since the Feynman propagator $\Delta_F(x)$ is symmetric under $x \mapsto -x$. So we will keep arrows on lines where they contain any information about which of the two adjacent vertices belonged to $\phi$ and $\phi^*$ and drop them everywhere else, that is where the diagrams represents a sum over contributions with arrows in different directions. \\


\textbf{Example 1:} Let
\begin{equation}
	\label{eq:bajo_komm}
	P_1 = \phi^2 \qquad \text{and} \qquad P_2 = (\phi^*)^2 \; .
\end{equation}
Only the derivative $\theta$ contributes to the MWI, since there are no derivatives present in $P_1,P_n$, so applying $\theta_{\mu}$ always yields zero. The MWI reads 
\begin{align}
	\partial_{\mu}^y t \big(\phi^2,(\phi^*)^2,j^{\mu}\big) &(x_1-y,x_2-y) \notag\\
	&= 2 \hbar \Big( \delta(y-x_1) - \delta(y-x_2) \Big) \cdot 
	t  \big(\phi^2,(\phi^*)^2\big) (x_1-x_2) \; .
\end{align}
Computing the $t$-products amounts to finding all possibilities of completely contracting the terms in equations (\ref{eq:bajo_komm}) with each other and finding the combinatorial factors -- that is the number of contraction schemes that lead to the same diagram. The result can be represented diagrammatically as \\

\begin{align}
	i \partial_{\mu}^y\; 
	\Bigg(\quad
	\parbox[h][0.10\linewidth][c]{0.2\linewidth}{
		\includegraphics[width=\linewidth]{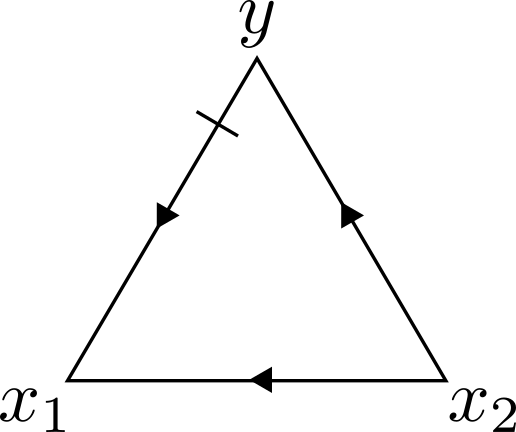}
	}
	\quad &-\quad
	\parbox[h][0.10\linewidth][c]{0.2\linewidth}{
		\includegraphics[width=\linewidth]{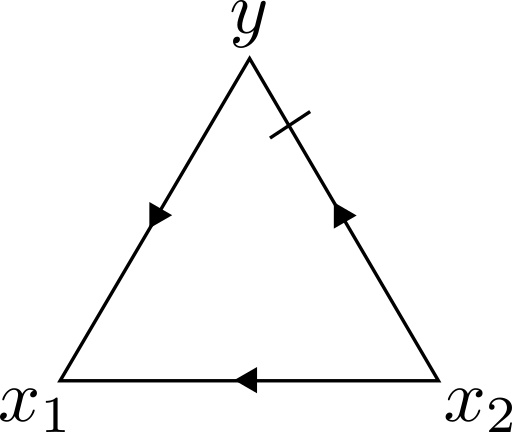}
	}
	\quad\Bigg) \notag\\
	\notag \\ \notag\\
	= \delta(y-x_1) \;\cdot\;
	\parbox[h][0.10\linewidth][c]{0.2\linewidth}{
		\includegraphics[width=\linewidth]{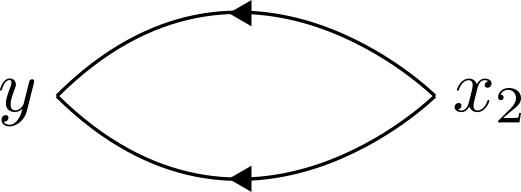}
	} 
	\quad&-\quad 
	\delta(y-x_2) \;\cdot\;
	\parbox[h][0.10\linewidth][c]{0.2\linewidth}{
		\includegraphics[width=\linewidth]{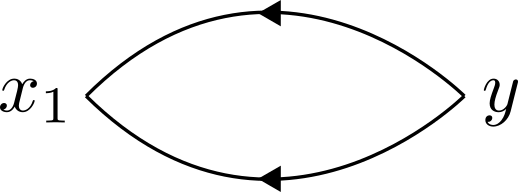}
	} \quad .
	\notag
\end{align}


\textbf{Example 2:} Let
\begin{equation}
	P_1 = \phi^* \phi^2 \qquad \text{and} \qquad P_2 = \phi\, (\phi^*)^2 \; .
\end{equation}
Again, no derivatives are present among the $P_i$. The corresponding MWI is
\begin{align}
	\partial_{\mu}^y t \big(\phi^* \phi^2,\phi\, (\phi^*)^2,j^{\mu}\big) &(x_1-y,x_2-y) \notag\\
	&=  \hbar\, \Big( \delta(y-x_1) - \delta(y-x_2) \Big) \cdot 
t  \big(\phi^* \phi^2,\phi\, (\phi^*)^2\big) (x_1-x_2) \; ,
\end{align}
which gives the following diagrams (in the first line we drop the arrows):

\begin{align}
	\notag\\
	i \partial_{\mu}^y \; 
	\Bigg(\quad
	\parbox[h][0.10\linewidth][c]{0.2\linewidth}{
		\includegraphics[width=\linewidth]{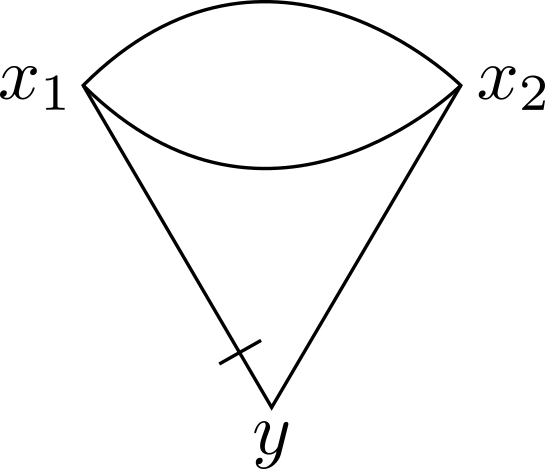}
	}
	\quad &-\quad
	\parbox[h][0.10\linewidth][c]{0.2\linewidth}{
		\includegraphics[width=\linewidth]{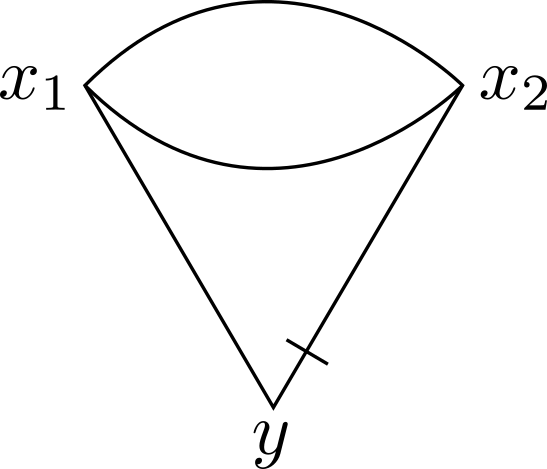}
	}
	\quad\Bigg) \notag\\
	\notag \\
	= \delta(y-x_1) \;\cdot\;
	\parbox[h][0.10\linewidth][c]{0.2\linewidth}{
		\includegraphics[width=\linewidth]{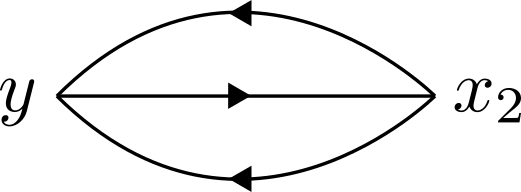}
	} 
	\quad&-\quad 
	\delta(y-x_2) \;\cdot\;
	\parbox[h][0.10\linewidth][c]{0.2\linewidth}{
		\includegraphics[width=\linewidth]{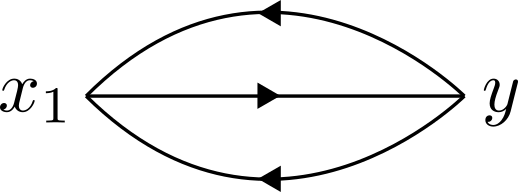}
	} \quad . 
	\notag
\end{align}


\textbf{Example 3:} Let
\begin{equation}
P_1 = \phi^* \phi \qquad \text{and} \qquad P_2 = j^{\nu} \; .
\end{equation}
In this case, there are derivatives in $P_2$ that via $\theta_{\mu}$ contribute to the total divergence. Since both $P_1$ and $P_2$ have total charge number zero, this time the contribution from $\theta$ vanishes. We get 
\begin{align}
	\partial_{\mu}^y t \big(\phi^* \phi, j^{\nu},j^{\mu}\big) &(x_1-y,x_2-y) \notag\\
	&= 2\,i \hbar\, \partial_y^{\nu} \delta(y-x_2) \cdot 
	t  \big(\phi^* \phi,\phi^* \phi\big) (x_1-x_2)
\end{align}
which can diagramatically be expressed as\\

\begin{align}
	\partial_{\mu}^y \; 
	\Bigg(\quad&
	\parbox[h][0.10\linewidth][c]{0.2\linewidth}{
		\includegraphics[width=\linewidth]{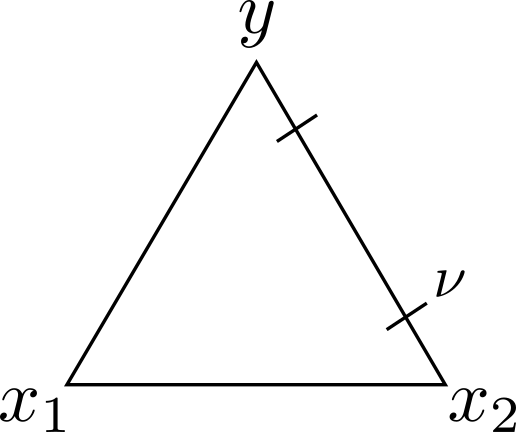}
	}
	\quad + \quad
	\parbox[h][0.10\linewidth][c]{0.2\linewidth}{
		\includegraphics[width=\linewidth]{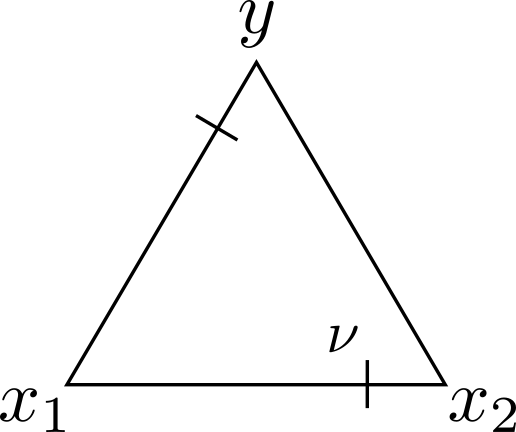}
	}
	\quad - \quad
	\parbox[h][0.10\linewidth][c]{0.2\linewidth}{
		\includegraphics[width=\linewidth]{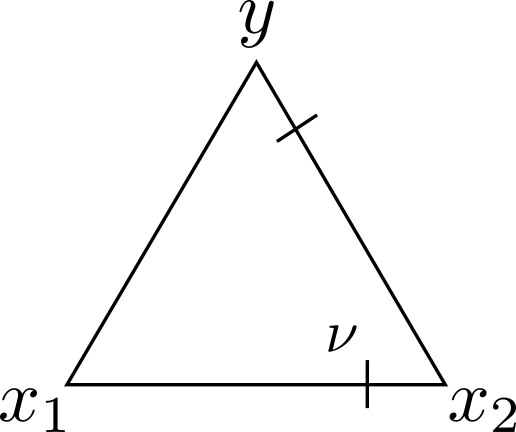}
	}
	\notag\\\notag\\\notag\\
	\quad - \quad&
	\parbox[h][0.10\linewidth][c]{0.2\linewidth}{
		\includegraphics[width=\linewidth]{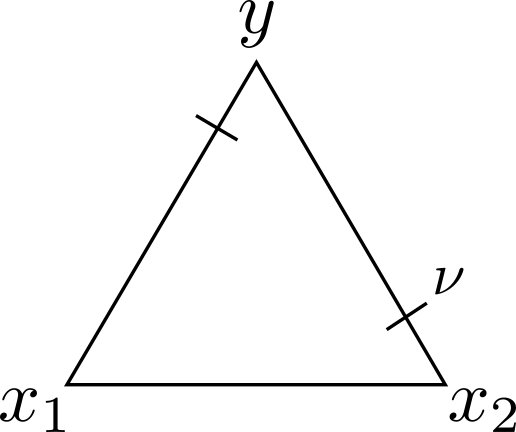}
	}
	\quad\Bigg)
	\qquad = \qquad 
	i \,\partial_y^{\nu} \; \delta(y-x_2) \;\cdot\;
	\parbox[h][0.10\linewidth][c]{0.2\linewidth}{
		\includegraphics[width=\linewidth]{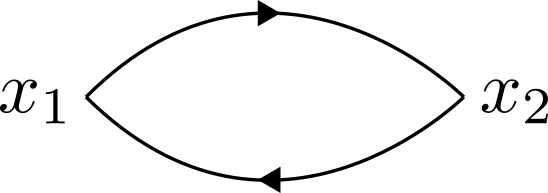}
	} \quad \; .
	\notag
\end{align}

\vspace{2cm}

\textbf{Example 4:} Let
\begin{equation}
	P_1 = (\phi^*\phi)^2 \;, \qquad P_2 = \phi^*(\phi)^2 \;, \qquad \text{and}\qquad  P_3 =   (\phi^*)^2\phi \; .
\end{equation}
This time we have three field polynomials, but none of them contains derivatives. We get 
\begin{align}
	\partial_{\mu}^y\; t &\big((\phi^*\phi)^2, \phi^*(\phi)^2 ,(\phi^*)^2\phi, j^{\mu}\big) (x_1-y,x_2-y,x_3-y) \notag\\
	&=  \hbar \; \Big( \delta(y-x_2) - \delta(y-x_3) \Big) \cdot 
	t  \big((\phi^*\phi)^2,\phi^*(\phi)^2,(\phi^*)^2\phi\big) (x_1-x_3,x_2-x_3) \; ,
\end{align}
which when calculated can be expressed as

\newpage 

\begin{align}
	\notag\\
	i \partial_{\mu}^y \; 
	\Bigg(
	\; - \;
	2 \quad &
	\parbox[b][0.050\linewidth][c]{0.2\linewidth}{
		\includegraphics[width=\linewidth]{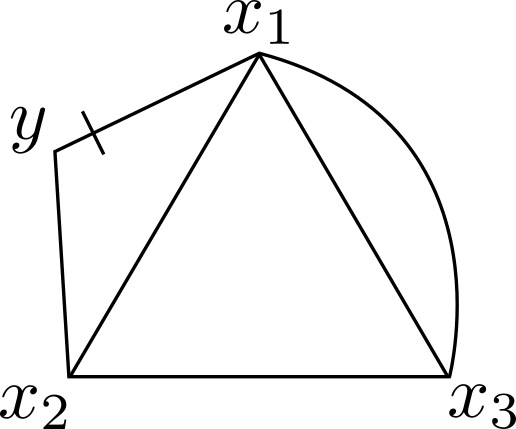}
	}
	\quad + \;
	2 \quad 
	\parbox[b][0.050\linewidth][c]{0.2\linewidth}{
		\includegraphics[width=\linewidth]{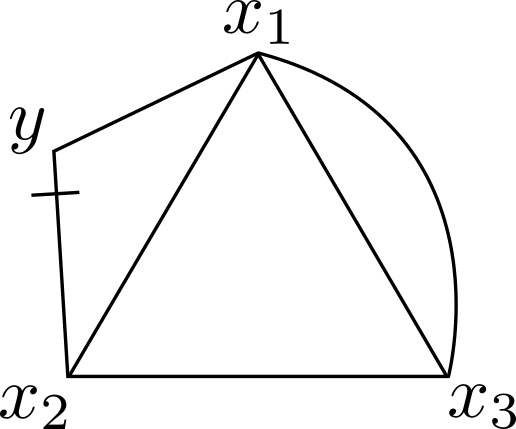}
	}
	\quad - \quad
	\parbox[t][0.002\linewidth][c]{0.2\linewidth}{
		\includegraphics[width=\linewidth]{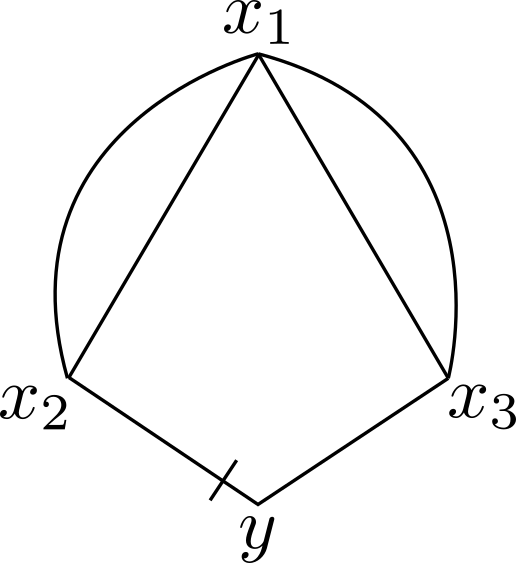}
	}
	\notag\\\notag\\\notag\\\notag\\\notag\\
	\quad + \quad&
	\parbox[t][0.002\linewidth][c]{0.2\linewidth}{
		\includegraphics[width=\linewidth]{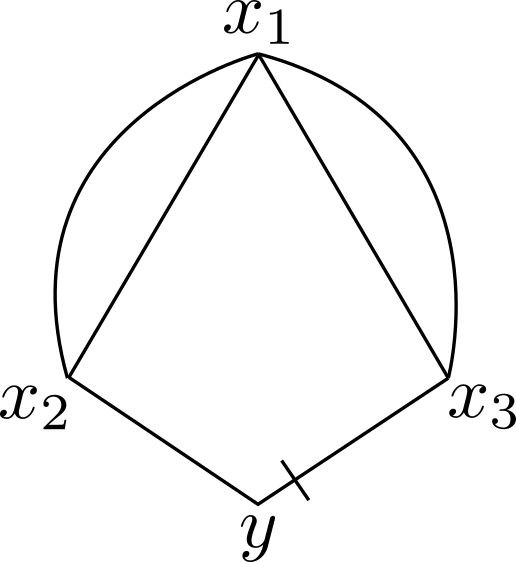}
	}
	\quad - \quad
	2
	\parbox[b][0.050\linewidth][c]{0.2\linewidth}{
		\includegraphics[width=\linewidth]{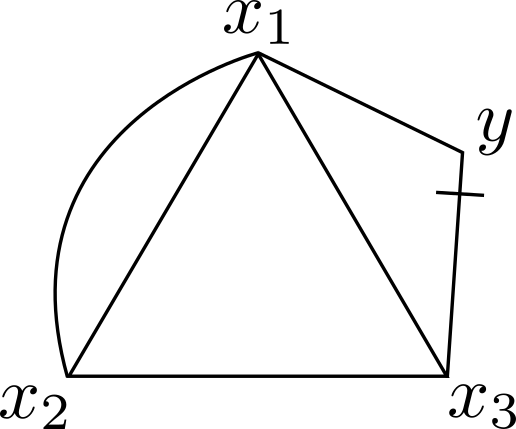}
	}
	\quad + \quad
	2
	\parbox[b][0.050\linewidth][c]{0.2\linewidth}{
		\includegraphics[width=\linewidth]{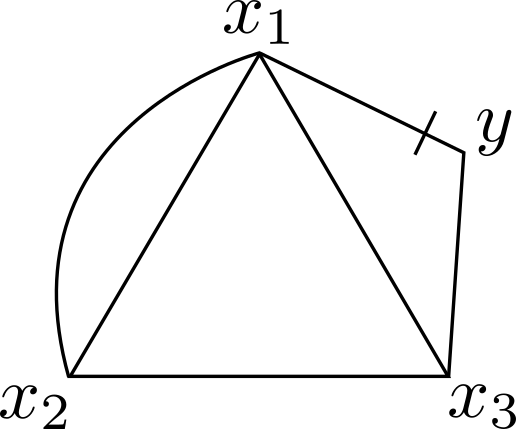}
	}
	\quad\Bigg)
	\notag\\\notag\\\notag\\\notag\\\notag\\
	 = \;\qquad 
	& \delta(y-x_2) \;\cdot\; 5 \;
	\parbox[b][0.035\linewidth][c]{0.2\linewidth}{
		\includegraphics[width=\linewidth]{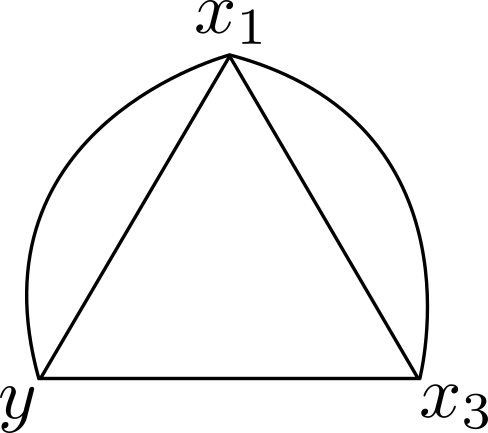}
	} 
	- \quad
	\delta(y-x_3) \;\cdot\; 5 \, 
	\parbox[b][0.035\linewidth][c]{0.2\linewidth}{
		\includegraphics[width=\linewidth]{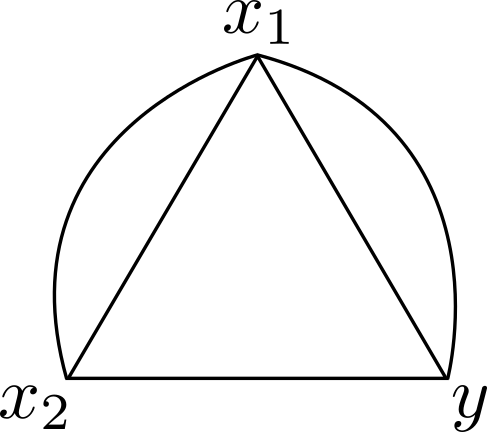}
	} \quad . 
	\notag
\end{align}

\vspace{2cm}

\begin{remark}
	The relations calculated in this section can be checked for the massless case at non-coinciding points by using $\dalemb \Delta_F(x) = -i \delta(x)$ for $x \neq 0$ and $m=0$. 
\end{remark}

	\addchap{Conclusion}
	\clearpairofpagestyles
\ohead{Conclusion}
\KOMAoptions{
	headsepline = true}
\cfoot{\pagemark}

The goal of this thesis was to prove that the MWI for the complex scalar field with quartic interaction can be satisfied to all orders of the perturbative expansion. This was achieved for a certain class of arguments of the involved $T$-products, namely the interaction $L$, the current $j^{\mu}$ and all their submonomials, which are enough to express all the $T$-products of the form $T(L,\dots,L,j\dots,j)$ via their causal Wick expansion. These, in turn, are all the $T$-products needed to establish the preservation of symmetries in the presence of interactions of degrees up to the quartic one that do not contain derivatives. Unexpectedly, a new term which has the form of a total divergence appears in the scalar MWI. This term arises when we allow fields as arguments that contain derivatives of the basic fields. This is the case for the current $j^{\mu}$ which is of physical interest, hence the additional term has to be taken into account. Our original contribution is to carry this term through the steps of the proof when adapting the QED case in chapter \ref{chap:proof} and to make the necessary modifications. In the last step of the proof, a case distinction has to be performed which deviates essentially from the QED case. \\

The framework of deformation quantization combined with causal perturbation theory gives a clear account of the relations between the classical and the quantum theory. It furthermore makes precise the notion of renormalization, which can be performed as the extension of distributions without having to deal with divergent quantities. By formulating symmetries through the MWI it is also possible to clarify the relation between quantum and classical symmetries. Introducing this formalism allowed us to show that the relevant symmetries described by the MWI for the classical complex scalar field carry over into the quantum world, giving a proper proof comprising well defined objects. Nonetheless our approach is a perturbative one, involving formal power series without a notion of convergence. This allows to deal with the scalar field with quartic interactions, which is subject to non-linear field equations. \\

A way to further investigate on the complex scalar MWI could be to consider higher order interactions than the quartic one, which would lead to more possible cases in the last step of the proof of our main theorem. It is not clear whether at some stage anomalies occur that can no longer be removed by an admissible renormalization. Another line of inquiry would be to consider quantum field theories on gravitational backgrounds, that is on a curved spacetime. Causal perturbation theory can be made to work on a certain class of curved spacetimes (see e.g \cite{fred00}), so our framework would carry over. However, the translation invariance of VEVs would get lost, on which our proof of the MWI heavily relied on. So to compensate for this, one would have to look for major modifications in essential parts of the proof. \\

Finally, the result of this thesis can be used to prove gauge invariance of scalar QED, which amounts to additionally considering the photon field $A^{\mu}$ and a new kind of interaction. However, the conserved current of scalar QED is the same as for the complex scalar case, so fields involving first derivatives need to be considered as arguments of the $T$-products involved in the relevant MWI. This gauge invariance of scalar QED will be proven in an upcoming paper together with Michael Dütsch and Karl-Henning Rehren.  \\

\newpage

	\addchap{Appendix}
	\appendix

\clearpairofpagestyles
\ohead{Appendix}
\KOMAoptions{
	headsepline = true}
\cfoot{\pagemark}

\section*{A.1 Minkowski space}

We use the standard definitions for Minkowski space and the light cones. 

\begin{manualdefinition}{A.1.1}
Let $\mathbb{M}= \mathbb{M}_d$ where $d>2$ be the $d$-dimensional Minkowski space with metric 
\begin{equation*}
	g = \text{diag}(+,-\dots,-) \; .
\end{equation*}
The forward and backward light cones are defined as 
\begin{equation*}
	V_+ := \{x \in \mathbb{M} \; | \; x^2 > 0, x^0 > 0\} \;, \quad 
	V_- := \{x \in \mathbb{M} \; | \; x^2 > 0, x^0 < 0\} \;. 
\end{equation*}
Let furthermore $\mathscr{L}_+^{\uparrow}$ be the proper, ortochronous Lorentz group and $\mathscr{P}_+^{\uparrow}$ the corresponding Poincaré group. 
\end{manualdefinition}

\section*{A.2 Vector spaces of formal power series}

\renewcommand{\theequation}{A.2.1}

\begin{manualdefinition}{A.2.1}
	\label{def:formal_po}
	Let $\mathscr{V}$ be a vector space and $\lambda \in  \mathbb{R} / \{0\}$. The vector space of \textit{formal power series} in $\lambda$ with coefficients in $\mathscr{V}$ is the set 
	\begin{equation}
		\mathscr{V}[\![\hbar]\!] := 
		\big\{ V \equiv \sum_{n=0}^{\infty} \left.  V_n \lambda^n \equiv (V_n)_{n\in\mathbb{N}}  \; \right\vert \; V_n \in \mathscr{V}
		\big\}
	\end{equation}
	with the addition and scalar multiplication 
	\begin{equation*}
		(V+cW)_n := V_n + cW_n \; ,\quad c \in\mathbb{C} \; .
	\end{equation*}
	If $\mathscr{V}$ is a unital $*$-algebra, then $\mathscr{V}[\![\hbar]\!]$ is also by using the obvious definitions for the multiplication and the $*$-operation. 
\end{manualdefinition}
So the infinite sum used when writing elements of spaces of formal power series is merely a notational convention. The elements of such spaces should -- from a mathematical point of view -- be considered to be sequences, not series.

\section*{A.3 Wave front sets and products of distributions}

The idea of wave front sets is to use the characterization of smoothness of functions by the fast decrease of their Fourier transform to describe the directions in which the singularities of a distribution are localized. It is introduced by the following definitions, which we take from \cite{brou14}. 

\begin{manualdefinition}{A.3.1}
	A \textit{conical neighbourhood} of a point $k \in \mathbb{R}^n / \{0\}$ is a set $V \subset \mathbb{R}^n$ such that $V$ contains the ball $B_{\epsilon}(k):= \{q\in\mathbb{R}^n \; | \; |q-k| < \epsilon\}$ for some $\epsilon>0$ and, for all $p\in V$ and all $\alpha >0$, $\alpha\cdot p \in V$. 
\end{manualdefinition}

\begin{manualdefinition}{A.3.2}
	A smooth function $g \in \mathcal{C}^{\infty}(\mathbb{R}^n)$ is said to be \textit{fast decreasing} on a conical neighbourhood $V$ if, for any $N \in \mathbb{N}$, there is a constant $C_N$ such that $|g(q)| \leq C_N (1+|q|)^{-N}$ for all $q\in V$. 
\end{manualdefinition}

\begin{manualdefinition}{A.3.3}
	For a distribution $u \in \mathscr{D}'(\mathbb{R}^n)$, a point $(x,k)\in \mathbb{R}^n \times (\mathbb{R}^n/\{0\})$ is called a \textit{regular directed point} of $u$ if and only if there exists 
	\begin{itemize}
		\item a function $f \in \mathscr{D}(\mathbb{R}^n)$ with $f(x)=1$ and 
		\item a closed conical neighbourhoof $V \in \mathbb{R}^n$ of $k$, such that $\widehat{fu}$ is fast decreasing on $V$. 
	\end{itemize}
\end{manualdefinition}

\begin{manualdefinition}{A.3.4}
	\label{def:wvs}
	The \textit{wave front set} of a distribution $u \in \mathscr{D}'(\mathbb{R}^n)$ is the set, denoted by $\text{WF}(u)$, of points $(x,k) \in \mathbb{R}^n \times (\mathbb{R}^n/\{0\})$ which are not regular directed for $u$. 
\end{manualdefinition}

This set can be used to give conditions on when the pointwise product of distributions can be meaningfully defined. 

\begin{manualtheorem}{A.3.5}
	\label{thm:hoerm}
	Let $u$ and $v$ be distributions in $\mathscr{D}'(\mathbb{R}^n)$. Assume that there is no point $(x,k)$ in $\text{WF}(u)$ such that $(x,-k)$ belongs to $\text{WF}(v)$. Then the product $uv$ can be defined as the pullback of their tensor product along the diagonal map $D: \mathbb{R}^n \rightarrow \mathbb{R}^n \times \mathbb{R}^n$ \normalfont{:}
	\begin{equation}
		u \cdot v := D^*(u\otimes v) \; .
	\end{equation}
\end{manualtheorem}

	\newpage

	\printbibliography[
	heading=bibintoc,
	title={Bibliography}
	]

\end{document}